\documentclass[12pt]{article}

\usepackage[a4paper]{geometry}
\usepackage{amssymb}
\usepackage{amsmath}
\usepackage{amsthm}

\newcounter{constants}
\newcommand{\C}{\refstepcounter{constants}C_{\theconstants}}
\newcommand{\D}[1]{\addtocounter{constants}{-#1}C_{\theconstants}\addtocounter{constants}{#1}}
\newcommand{\CI}{\setcounter{constants}{0}\refstepcounter{constants}C_{\theconstants}}
\renewcommand{\Re}{\operatorname{Re}}
\renewcommand{\Im}{\operatorname{Im}}

\theoremstyle{plain}
\newtheorem{theorem}{Theorem}[section]
\newtheorem{proposition}[theorem]{Proposition}
\newtheorem{lemma}[theorem]{Lemma}
\newtheorem{corollary}[theorem]{Corollary}

\theoremstyle{definition}

\newtheorem{assumption}[theorem]{Assumption}

\theoremstyle{remark}
\newtheorem{remark}[theorem]{Remark}
\newtheorem{remarks}[theorem]{Remarks}

\numberwithin{equation}{section}

\title{Low-energy resolvent estimates for slowly decaying attractive potentials}

\author{Kenichi {\scshape Ito}\footnote{Department of Mathematics, Graduate School of Science, Kobe University,
1-1, Rokkodai, Nada-ku, Kobe 657-8501, Japan.
E-mail: \texttt{ito-ken@math.kobe-u.ac.jp}. 
}
\ and
Tomoya {\scshape Tagawa}\footnote{Graduate School of Mathematical Sciences, 
The University of Tokyo, 3-8-1 Komaba, Meguro-ku, Tokyo 153-8914, Japan.
E-mail: \texttt{tagawa-tomoya1212@g.ecc.u-tokyo.ac.jp}. 
}
}

\date{}

\begin{document}
\allowdisplaybreaks
\maketitle

\begin{abstract}
We discuss the low-energy resolvent estimates 
for the Schr\"odinger operator with slowly decaying attractive potential.
The main results are Rellich's theorem, the limiting absorption principle and Sommerfeld's uniqueness theorem.  
For the proofs, we employ an elementary commutator method due to Ito--Skibsted,
for which neither of microlocal or functional analytic techniques is required.
\end{abstract}



\medskip

\tableofcontents

\section{Introduction}\label{2411241}

In the present paper, we discuss the uniform low-energy resolvent estimates 
for the Schr\"odinger operator on $\mathbb R^d$ with $d\in \mathbb N=\{1,2,\ldots\}$: 
\begin{align*}
H=-\tfrac12\Delta+V+q
.
\end{align*}
The operator $\Delta$ is the ordinary Laplacian, and we shall often write it as 
\[
-\Delta=p^2=p_ip_i=p\cdot p,\quad p=-\mathrm i\nabla=-\mathrm i\partial.\]
Here and below, the Einstein summation convention is adopted without tensorial superscripts. 
The potential $V$ is \textit{slowly decaying} and \textit{attractive},
and $q$ is a perturbation of short-range type for $-\tfrac12\Delta+V$. 
Precise assumptions will be given soon below in Assumption~\ref{cond:230806}. 

The scattering theory for slowly decaying potentials, 
whether attractive or repulsive, was first studied by Yafaev \cite{Y}, 
and was developed by Nakamura \cite{Nakam} and Fournais--Skibsted \cite{FS} 
in connection with the threshold resonances. 
These approaches were even further sophisticated by Skibsted \cite{S}, 
where the full stationary scattering theory was carried out
with sharp estimates involving long-range perturbations. 
In addition, recently, Sussman~\cite{Su} has discussed it on an asymptotically Eulidean manifold. 
Our goal is very similar to that of \cite{S}, but we would like to relax the smoothness assumption
and simplify the proofs. 
In fact, almost all of these former results impose $C^\infty$ smoothness on the potential, 
since they are dependent on the pseudodifferential techniques. 
Recently, Ito--Skibsted \cite{IS} have developed simple commutator arguments 
to discuss the stationary scattering theory for strictly positive energies. 
We are going to follow their approach for the low energies including $0$. 

The setting of the paper is in part stronger than it should be, in that the principal part of the potential has 
spherical symmetry, 
and in that the perturbation is of short-range type, either of which were not assumed in \cite{S}. 
However, in the present paper we can relax smoothness to $C^2$ with a new technique from \cite{IS}. 
We expect that the paper could be extended to the setting of \cite{S}, 
if we employed a method of an even recent paper by Ito--Skibsted \cite{IS2}.
This would require much longer and more complicated preliminaries,
and for this reason, here we have decided to present only simple start-up arguments for the subject.
Hopefully, we could discuss such an extension elsewhere. 


If the leading part of the potential is spherically symmetric, 
one can extract precise information on the scattering matrix, see, e.g.,\ \cite{AK}. 
In particular,  Derezi\'nski--Skibsted \cite{DS1,DS2} and Frank \cite{F} 
computed its FIO expression at energy zero. 
The uniform resolvent estimates for slowly decaying potentials have an application to the Strichartz estimates, 
see \cite{M} and \cite{T} for the repulsive case. 
Mochizuki--Nakazawa \cite{MN2} also investigated slowly decaying potential in an exterior domain.

\subsection{Basic settings}

\subsubsection{Slowly decaying attractive potential}

Throughout the paper, we use notation $\langle x\rangle=(1+x^2)^{1/2}$.

\begin{assumption}\label{cond:230806}
Let $V\in C^2(\mathbb R^d)$ be spherically symmetric, 
and there exist $\nu,\epsilon\in (0,2)$ and $c,C>0$ such that for any $|\alpha|\le 2$ and $x\in\mathbb R^d$
\[
|\partial^\alpha V(x)|\le C\langle x\rangle^{-\nu-|\alpha|},\quad 
V(x)\le -c\langle x\rangle^{-\nu},\quad 
x\cdot (\nabla V(x))\le -(2-\epsilon)V(x).
\]
In addition, let $q\in L^\infty(\mathbb R^d)$, and there exist $\nu'\in (\nu,2]$ and $C'>0$ such that for any $x\in\mathbb R^d$ 
\[
|q(x)|\le C'\langle x\rangle^{-1-\nu'/2}
. 
\]
\end{assumption}

\begin{remarks}
\begin{enumerate}
\item
We shall often write simply
\[r=|x|,\quad \partial_r=r^{-1}x\cdot \nabla 
,\quad 
x\cdot \nabla V=r\partial_r V,
\]
and also abuse notation as $V(x)=V(r)$. 
\item
Assumption~\ref{cond:230806} is not directly comparable with those in the previous works \cite{Y,Nakam,FS,S},
but we would like to emphasize that our $V$ is only $C^2$. 
Such a relaxation is possible since we avoid the microlocal techniques. 

\item
The spherical symmetry of $V$ is more or less necessary if we finally aim at the stationary scattering theory, 
which we shall not discuss in the paper. 
In fact, it excludes logarithmic spirals in the classical orbits. 
See \cite{Y,S} for the related results. 
\cite{S} is much more general in this aspect, including long-range perturbations, although in the $C^\infty$ settings. 
On the other hand, such symmetry was unnecessary in \cite{Nakam,FS} since their goals are different. 

\item
Under Assumption~\ref{cond:230806}, $H$ is self-adjoint on the Hilbert space $\mathcal H=L^2(\mathbb R^d)$.
The self-adjoint realization is denoted by the same notation $H$.
\end{enumerate}
\end{remarks}

\subsubsection{Agmon--H\"ormander spaces}

We define an \textit{effective time}, or an \textit{escape function}, as 
\begin{align}
\tau(\lambda,x)=\int_0^ra(\lambda,s)^{-1}\,\mathrm ds;
\quad 
a(\lambda,r)=(2\max\{\lambda,0\}-2V(r))^{1/2}
\label{eq:230820}
\end{align}
for $(\lambda,x)\in \mathbb R \times\mathbb R^d$.
See Proposition~\ref{prop:23080720} for a motivation of these terminologies. 
Here we only remark that there exist $c,C>0$ such that 
\begin{align}
c\langle \lambda\rangle^{-1/2}r\langle r\rangle^{\nu/2}\le \tau(0,x)\le C\langle \lambda\rangle^{-1/2}r\langle r\rangle^{\nu/2}
,
\label{eq:23111912}
\end{align}
and that for any $\lambda_0>0$ there exist $c',C'>0$ such that 
\begin{align}
c'\langle \lambda\rangle^{-1/2}r\le \tau(\lambda,x) \le C'\langle 
\lambda\rangle^{-1/2}r\ \ \text{uniformly in }\lambda\ge \lambda_0. 
\label{eq:23111913}
\end{align}

Using the function $\tau$, we introduce the associated \emph{Agmon--H\"ormander spaces},
a variant of the Besov spaces, as 
\begin{align*}
\mathcal B(\lambda)&=
\bigl\{\psi\in L^2_{\mathrm{loc}};\ \|\psi\|_{\mathcal B(\lambda)}<\infty\bigr\},\quad 
\|\psi\|_{\mathcal B(\lambda)}=\sum_{n\in\mathbb N} 2^{n/2}\|1_n(\lambda)\psi\|_{\mathcal H},\\
\mathcal B^*(\lambda)&=
\bigl\{\psi\in L^2_{\mathrm{loc}};\ \|\psi\|_{\mathcal B^*(\lambda)}<\infty\bigr\},\quad 
\|\psi\|_{\mathcal B^*(\lambda)}=\sup_{n\in\mathbb N}2^{-n/2}\|1_n(\lambda)\psi\|_{\mathcal H},
\\
\mathcal B^*_0(\lambda)
&=
\Bigl\{\psi\in \mathcal B^*(\lambda);\  \lim_{n\to\infty}2^{-n/2}\|1_n(\lambda)\psi\|_{\mathcal H}=0\Bigr\}
,
\end{align*}
where we let 
\begin{align*}
\begin{split}
1_1(\lambda)&=1\bigl(\bigl\{x\in\mathbb R^d;\ \tau(\lambda,x)<2\bigr\}\bigr),\\
1_n(\lambda)&=1\bigl(\bigl\{x\in\mathbb R^d;\ 2^{n-1}\le \tau(\lambda,x)<2^n\bigr\}\bigr)
\ \  \text{for }n=2,3,\ldots 
\end{split}
\end{align*}
with $1(S)$ being the sharp characteristic function of a subset $S\subset \mathbb R^d$. 
Note by \eqref{eq:23111912} and \eqref{eq:23111913} that, if we define the \emph{weighted $L^2$ spaces} as  
\begin{equation*}
L^{2}_{s,\lambda}=\langle \tau \rangle^{-s}\mathcal H\ \ \text{for }s\in\mathbb R 
,
 \end{equation*} 
depending on $\lambda$ is through $\tau$, 
then for $\lambda \in \mathbb{R}$ and any $s>1/2$ 
\begin{equation*}
 L^2_{s,\lambda}\subsetneq \mathcal B(\lambda)\subsetneq L^2_{1/2,\lambda}
 \subsetneq \mathcal H
\subsetneq L^2_{-1/2,\lambda}\subsetneq \mathcal B^*_0(\lambda)\subsetneq \mathcal
 B^*(\lambda)\subsetneq L^2_{-s,\lambda}.
\end{equation*}
We also note, for any compact interval $I\subset (0,\infty)$, 
the spaces $\mathcal B(\lambda)$ are identical for all $\lambda\in I$ along with uniformly equivalent norms. 
The same is true for $\mathcal B^*(\lambda)$ with $\lambda\in I$.

\subsection{Main results}

Now we present a series of the main results of the paper. 

\subsubsection{Rellich's theorem}

We start with \textit{Rellich's theorem}, or absence of generalized eigenfunctions in $\mathcal B^*_0$. 
It is a basis of our theory and will be repeatedly referred to in the latter part of the paper. 

\begin{theorem}\label{thm:230607}
Suppose Assumption~\ref{cond:230806}. If $\lambda\ge 0$ and  $\phi\in \mathcal B^*_0(\lambda)$ satisfy 
\[(H-\lambda)\phi=0\ \ \text{in the distributional sense,}\]
then $\phi\equiv 0$. 
In particular, the self-adjoint realization of $H$ on $\mathcal H$ does not have nonnegative eigenvalues:
$\sigma_{\mathrm{pp}}(H)\cap [0,\infty)=\emptyset$. 
\end{theorem}

\begin{remark}
For a positive spectral parameter $\lambda>0$, the result is already known under a more general assumption, see ,e.g.,\ \cite{IS}. 
For $\lambda=0$, a similar result can be deduced from \cite{S}, but the setting and the proof are different from ours.  
\end{remark}

\subsubsection{LAP bounds}

We next present the \textit{LAP (limiting absorption principle) bounds} for the resolvent
\[R(z)=(H-z)^{-1}\in\mathcal L(\mathcal H)\ \ \text{for }z\in\rho(H). \]
Let us introduce 
\begin{align}
A=\mathrm i [H,\tau]=\mathop{\mathrm{Re}}(a^{-1}p_r)
;\quad 
p_r=-\mathrm i\partial_r
,
\label{eq:2308243}
\end{align}
and 
\begin{align}
L=p_i\ell_{ij}p_j;\quad 
\ell_{ij}=\delta_{ij}-r^{-2}x_ix_j=r(\nabla^2r)_{ij}
,
\label{eq:23030414}
\end{align}
where $\delta$ is the Kronecker delta. 
Note that for $x\neq 0$ the matrix $\ell$ represents the orthogonal projection onto the spherical direction. 
Set for any $\rho>0$ and $\omega\in (0,\pi)$ 
\[
\Gamma_{\pm}(\rho,\omega)
=\bigl\{z\in\mathbb C;\ 0<|z|<\rho,\ 0<\pm\arg z<\omega\bigr\}, 
\]
respectively. 
In addition, for an operator $T$ on $\mathcal H$ we denote $\langle T\rangle_\psi=\langle\psi,T\psi\rangle_{\mathcal H}$.

\begin{theorem}\label{thm:2308271915}
Suppose Assumption~\ref{cond:230806}, and let $\rho>0$ and $\omega\in (0,\pi)$.
Then there exists $C>0$ such that for any $\phi=R(z)\psi$ with $z=\lambda \pm \mathrm i\mu \in \Gamma_{\pm}(\rho,\omega)$ and $\psi\in \mathcal B$
\[
\|\phi\|_{\mathcal B^*(\lambda)}
+\|a^{-1}p_r\phi\|_{\mathcal B^*(\lambda)}
+\bigl\langle \phi, a^{-2}\langle \tau\rangle^{-1} L\phi\bigr\rangle^{1/2} 
\le C\|\psi\|_{\mathcal B(\lambda)}
,
\]
respectively. In particular, 
the self-adjoint realization of $H$ on $\mathcal H$ does not have nonnegative singular continuous spectrum: 
$\sigma_{\mathrm{sc}}(H)\cap [0,\infty)=\emptyset$. 
\end{theorem}
\begin{remarks}
\begin{enumerate}
\item
Theorem~\ref{thm:2308271915} provides a refinement of a part of the results from \cite{Nakam,FS},
while \cite{S} proved a similar result in a different setting with a different proof. 
\item
We can deduce uniform estimates in the high-energy regime $\lambda\to\infty$ by simple scaling arguments. 
We present it in Appendix~\ref{sec:231120} for completeness.
See also \cite{Nakam}. 
\end{enumerate}
\end{remarks}

\subsubsection{Radiation condition bounds and applications}

We next discuss the \textit{radiation condition bounds} and present their applications. 
We introduce an asymptotic complex phase $b$: For $z=\lambda \pm{i\mu} \in \mathbb{C} \backslash (-\infty,0]$,
\begin{equation}\label{eq:2312221350}
  b=b_{z}=\sqrt{2(z-V)} \mp{\mathrm i}\tfrac{\partial_{r}V}{4(z-V)}, 
\end{equation}
respectively.
Here we choose the branch of square root as $\Re \sqrt{w} >0$ for $w \in \mathbb{C} \backslash (-\infty,0].$ 
We also set for any $\rho \geq 0$
\[
\beta_{c,\rho}=\min \left\{\tfrac{2-\nu}{2(2+\nu)},\tfrac{\nu^{\prime}-\nu}{2+\nu}, \tfrac{2+3\epsilon}{8}\inf_{\lambda \in [0,\rho] } \left(\liminf_{|x| \to \infty }{\tau a r^{-1}}\right) \right\}.  \]
We note that $\beta_{c,\rho} >0$ for any $\rho >0$ by Lemma~\ref{lem:230925b}. 

\begin{theorem}\label{thm:2308271916}
  Suppose Assumption~\ref{cond:230806}, and let $\rho>0$ and $\omega\in (0,\pi)$.
  Then for all $\beta \in [0,\beta_{c,\rho})$, there exists $C>0$ such that for any $\phi=R(z)\psi$ with 
$z= \lambda \pm{i\mu} \in \Gamma_{\pm}(\rho,\omega)$ and $\psi\in \langle \tau \rangle^{-\beta}\mathcal B(\lambda)$
  \[
    \|\tau^{\beta} a^{-1}(p_{r} \mp{b_{z}})\phi\|_{\mathcal B^*(\lambda)}
    +\bigl\langle \phi, \tau^{2\beta} a^{-2}\langle \tau\rangle^{-1} L\phi\bigr\rangle^{1/2} 
    \le C\|\langle \tau \rangle^{\beta} \psi\|_{\mathcal B(\lambda)}.
  \]
\end{theorem}

Theorems~\ref{thm:2308271915} and \ref{thm:2308271916} imply the \textit{LAP}, 
or existence of the limiting resolvents $R(\lambda\pm\mathrm i0)$ of the following form.

\begin{corollary}\label{cor:2312221102}
  Suppose Assumption~\ref{cond:230806}, and let $\rho>0$ and $\omega\in (0,\pi)$.
  For any $s>1/2$ and $\gamma \in (0, \min\{s-1/2,\beta_{c,\rho}\})$, 
  there exists $C>0$ such that for any $z,z^{\prime} \in \Gamma_{+}(\rho,\omega)$ or any $z,z^{\prime} \in \Gamma_{-}(\rho,\omega)$
  \begin{equation}\label{eq:2312221103}
    \begin{split}
      \|R(z)-R(z^{\prime})\|_{\mathcal{B}\left(L^{2}_{s,\min\{\Re{z},\Re{z^{\prime}}\}}, L^{2}_{-s,\min\{\Re{z},\Re{z^{\prime}}\}}\right)}
    &\leq C |z-z^{\prime}|^{\gamma}, \\
    \|a^{-1}p_{r}R(z)-a^{-1}p_{r}R(z^{\prime})\|_{\mathcal{B}\left(L^{2}_{s,\min\{\Re{z},\Re{z^{\prime}}\}}, L^{2}_{-s,\min\{\Re{z},\Re{z^{\prime}}\}}\right)} 
    &\leq C |z-z^{\prime}|^{\gamma}.
    \end{split}
  \end{equation}
  In particular, the operators $R(z)$ and $a^{-1}p_{r}R(z)$ attain uniform limits 
  as $z\in \Gamma_{\pm}(\rho,\omega) \rightarrow \lambda \in [0,\rho)$
  in the norm topology of $\mathcal{B}(L^{2}_{s,\lambda}, L^{2}_{-s,\lambda})$, which are denoted by
  \begin{equation}\label{eq:2312221104}
    \begin{split}
    R(\lambda\pm{\mathrm i0}) &= \lim_{z \in \Gamma_{\pm}(\rho,\omega) \rightarrow \lambda} R(z) ,\\
    a^{-1}p_{r} R(\lambda \pm{\mathrm i0} ) &= \lim_{z \in \Gamma_{\pm}(\rho,\omega) \rightarrow \lambda} a^{-1}p_{r}R(z),
    \end{split}
  \end{equation} 
  respectively. These limits $R(\lambda \pm{\mathrm i0})$ and $a^{-1}p_{r}R(\lambda \pm{\mathrm i0})$ belong to $\mathcal{B}(\mathcal{B}(\lambda),\mathcal{B}^{*}(\lambda))$.
\end{corollary}

The radiation condition bounds extend to the limiting resolvents $R(\lambda\pm\mathrm i0)$. 
\begin{corollary}\label{cor:2312221105}
  Suppose Assumption~\ref{cond:230806}, and let $\rho>0$ and $\omega\in (0,\pi)$.
  Then for all $\beta \in [0,\beta_{c,\rho})$, there exists $C>0$ such that for any $\phi=R(\lambda\pm\mathrm i0)\psi$ with 
$\lambda \in[0,\rho)$ and $\psi \in \langle \tau \rangle^{-\beta} \mathcal B(\lambda)$
  \begin{align*}
    \|\tau^{\beta} a^{-1}(p_{r} \mp{b_{z}})\phi\|_{\mathcal B^*(\lambda)}
    +\bigl\langle \phi, \tau^{2\beta} a^{-2}\langle \tau\rangle^{-1} L\phi\bigr\rangle^{1/2} 
 \le C\|\langle \tau \rangle^{\beta} \psi\|_{\mathcal B(\lambda)}.
  \end{align*}
\end{corollary}

Finally we present Sommerfeld's uniqueness theorem, 
which characterizes the limiting resolvents $R(\lambda\pm\mathrm i0)$.
\begin{corollary}\label{cor:2312221107}
  Suppose Assumption~\ref{cond:230806}, and let $\lambda \geq 0$, $\phi \in L^{2}_{\mathrm{loc}}$ and $\psi \in \langle \tau \rangle^{-\beta} \mathcal{B}(\lambda)$ with $\beta \in [0,\beta_{c,\lambda})$
  Then $\phi = R(\lambda \pm{\mathrm i0}) \psi$ holds if and only if both of the following conditions hold:
  \begin{itemize}
    \item[$(\romannumeral1)$] $(H-\lambda) \phi = \psi$ in the distributional sense.
    \item[$(\romannumeral2)$] $\phi \in \langle \tau \rangle^{\beta} \mathcal{B}^{*}(\lambda)$ and $a^{-1}(p_{r}\mp{b_{\lambda}})\phi \in \langle \tau \rangle^{-\beta} \mathcal{B}^{*}_{0}(\lambda)$.
  \end{itemize} 
\end{corollary}

The paper is organized as follows. 
Section~\ref{sec:23082719} is devoted to the analysis of the corresponding classical mechanics,
which motivates our proofs. 
We prove Theorem~\ref{thm:230607} in Section~\ref{2411243}, 
Theorem~\ref{thm:2308271915} in Section~\ref{2411244}, 
and Theorem~\ref{thm:2308271916} and Corollaries~\ref{cor:2312221102}, \ref{cor:2312221105} and \ref{cor:2312221107}
in Section~\ref{2411245}. 
We discuss the uniform resolvent bounds for high energies in Appendix~\ref{sec:231120}.

\section{Classical mechanics}\label{sec:23082719}

In this section, we study the classical mechanics for the classical Hamiltonian 
\begin{align}
H(x,p)=\tfrac12p^2+V(x). 
\label{eq:23080719}
\end{align}
We hope, even though we use the same notation as in the quantum mechanics, there would be no confusion. 
Here we try to understand rolls of the effective time 
from a classical viewpoint. 
The results of the section are in fact not necessary for our purpose,
but would provide good motivations for the later arguments. 

For the arguments of the section, we can slightly relax the conditions on $V$. 

\begin{assumption}\label{cond:230806b}
Let $V\in C^1(\mathbb R^d)$ be spherically symmetric, 
and there exist $\nu,\epsilon\in (0,2)$ and $c>0$ such that for any $x\in\mathbb R^d$ 
\[
V(x)\le -c\langle x\rangle^{-\nu},\quad 
x\cdot (\nabla V(x))\le -(2-\epsilon)V(x). 
\]
\end{assumption}

Recall the Hamilton equations associated with \eqref{eq:23080719} are given by 
\[
\dot x= p,\quad \dot p=-\nabla V
.
\]
A classical orbit $(x,p)=(x(t),p(t))$ is \textit{forward/backward non-trapped} if 
\[
\lim_{t\to\pm\infty}|x(t)|=\infty,
\]
respectively.

\subsection{Effective time}

Here we discuss the order of scattering of the  classical particles. 
Let us start with a rough estimation. 
Let $(x(t),p(t))$ be a forward/backward non-trapped classical orbit of energy $\lambda=0$,
and suppose as $t\to\pm\infty$ 
\[x(t)\approx |t|^{\alpha},\quad p(t)\approx |t|^{\alpha-1}.\]
Then by  the law of conservation of energy, 
\[|t|^{2\alpha-2}\approx p(t)^2=-2V(x(t))\gtrsim |x(x)|^{-\nu}\approx |t|^{-\nu\alpha},\]
and this implies we should have $\alpha\ge 2/(2+\nu)$. 
Hence,
\[|x(x)|^{1+\nu/2}\gtrsim |t|\ \ \text{as }t\to\pm\infty,\]
and (the lower bound of) $x^{1+\nu/2}$ should play a roll of an \textit{effective time} for energy $\lambda=0$.  
We actually chose a more involved function $\tau$ from \eqref{eq:230820} for all the energy $\lambda\ge 0$,
and we can verify the corresponding property as follows.

\begin{proposition}\label{prop:23080720}
Suppose Assumption~\ref{cond:230806b}. 
If $(x(t),p(t))$ is a classical orbit of energy $\lambda\ge 0$ with 
\[
\pm p_r(0):=\pm |x(0)|^{-1}x(0)\cdot p(0)>0,\]
then for any $\pm t\ge 0$
\begin{align*}
\tau(\lambda,x(t))\ge  a(\lambda,x(0))^{-1} p_r(0)t+\tau(\lambda,x(0))
,
\end{align*}
respectively. 
\end{proposition}
\begin{proof}
Since 
\[
\tfrac{\mathrm d}{\mathrm dt}\tau(\lambda,x(t))=a(\lambda,x(t))^{-1} p_r(t);\quad 
p_r=(\nabla r)\cdot p=r^{-1}x\cdot p,\]
it suffices to show that 
\[
\tfrac{\mathrm d^2}{\mathrm dt^2}\tau(\lambda,x(t))\ge 0.
\]
Below we would like to avoid explicit $t$-derivatives to motivate the later stationary approach to the quantum mechanics. 
For that, we employ the Poisson brackets.
Let 
\begin{align}
D=\tfrac{\mathrm d}{\mathrm dt}=\{H,\cdot\},
\label{eq:230824}
\end{align}
and introduce a classical observable $A$ as 
\[
A=D\tau=\{H,\tau\}=  a^{-1}p_r
,
\]
see also \eqref{eq:2308243}. Then we compute 
\begin{align*}
D^2\tau
=\{H,A\}
&=
 a^{-3}(\partial_rV)p_r^2
+ a^{-1}p\cdot  (\nabla^2 r) p
- a^{-1}(\partial_r V) 
.
\end{align*}
Using Assumption~\ref{cond:230806}, the identity $\nabla r\otimes \nabla r+r(\nabla^2r)=\delta$, 
the law of conservation of energy and the convexity $\nabla^2 r\ge 0$, we can proceed as  
\begin{align*}
D^2\tau
&\ge 
  a^{-3}\bigl[
(\partial_rV)p_r^2
+(\partial_rV)p\cdot  (r\nabla^2 r) p
+(2\lambda-\epsilon V)p\cdot  (\nabla^2 r) p
-a^2 (\partial_r V) 
\bigr]
\\&
=
   a^{-3}\bigl[(2\lambda-\epsilon V)p\cdot  (\nabla^2 r) p
+2(\partial_rV)(H-\lambda)\bigr]
\\&
\ge 0.
\end{align*}
Hence, we are done. 
\end{proof}
To prove Proposition \ref{prop:23080720}, we show that the Poisson bracket between $H$ and $A$ is nonnegative.
In order to establish the main results, we analyze the quantization of the Poisson bracket between $H$ and $A$.

\section{Rellich's theorem}\label{2411243}

\subsection{Main propositions}

In this section, we prove Rellich's theorem, or Theorem~\ref{thm:230607}. 
The proof has two main steps, a priori super-exponential decay estimate 
and absence of super-exponentially decaying eigenfunctions. 
Their precise statements are the following. Throughout the section, we assume Assumption~\ref{cond:230806}.

\begin{proposition}\label{prop:23082418}
If $\phi\in \mathcal B^*_0(\lambda)$ with $\lambda\ge 0$ satisfies 
\[(H-\lambda)\phi=0\ \ \text{in the distributional sense,} \]
then $\mathrm e^{\alpha S}\phi\in \mathcal B^*_0(\lambda)$ for any $\alpha\ge 0$,
\end{proposition}

\begin{proposition}\label{prop:23082419}
If $\phi\in \mathcal B^*_0(\lambda)$ with $\lambda\ge 0$ satisfies 
\begin{enumerate}
\item
$(H-\lambda)\phi=0$ in the distributional sense,
\item
$\mathrm e^{\alpha S}\phi\in \mathcal B^*_0(\lambda)$ for any $\alpha\ge 0$,
\end{enumerate}
then $\phi\equiv 0$.  
\end{proposition}
Here, the function $S$ is defined by
\[
S(\lambda,x)=\int_0^r a(\lambda,s)\,\mathrm ds,
\qquad (\lambda,x)\in [0,\infty)\times\mathbb R^d,
\]
and solves the associated eikonal equation.
Theorem~\ref{thm:230607} follows immediately from these propositions. 
We will prove them in Sections~\ref{subsec:2308250} and \ref{subsec:2308251}, respectively,
somehow following the scheme provided in the proof of Proposition~\ref{prop:23080720}. 
In fact, we are going to compute and bound a \textit{distorted commutator}
\begin{align}
\mathop{\mathrm{Im}}(A\Theta (H-\lambda)). 
\label{eq:2308252}
\end{align}
Here the conjugate operator $A$ is from \eqref{eq:2308243}, 
and the weight function $\Theta$ is of the form  
\begin{align}
\Theta=\chi_{m,n}\mathrm e^{2\theta},\quad 
\theta=\alpha S+\beta\int_0^S(1+s/R)^{-1-\delta}\,\mathrm ds 
\label{eq:23082253}
\end{align}
with parameters $m,n\in\mathbb N_0$, $\alpha,\beta\ge 0$, $\delta>0$ and $R\ge 1$. 
The cutoff function $\chi_{m,n}$ is given as follows. 
Fix any $\chi\in C^\infty(\mathbb{R})$ such that 
\begin{align}
\chi(t)
=\left\{\begin{array}{ll}
1 &\mbox{ for } t \le 1, \\
0 &\mbox{ for } t \ge 2,
\end{array}
\right.
\quad
\chi'\le 0,
\label{eq:14.1.7.23.24}
\end{align}
and let $\chi_n,\bar\chi_n,\chi_{m,n}\in C^\infty(\mathbb R^d)$ be defined as 
\begin{align}
\chi_n=\chi(\tau/2^n),\quad \bar\chi_n=1-\chi_n,\quad
\chi_{m,n}=\bar\chi_m\chi_n.
\label{eq:11.7.11.5.14}
\end{align} 
We note the integral from \eqref{eq:23082253} is the so-called \textit{Yosida approximation}, so that 
it is bounded for each $R\ge 1$, and converges pointwise to $S$ as $R\to\infty$.

\subsection{A priori super-exponential decay estimate}\label{subsec:2308250}

Here we prove Proposition~\ref{prop:23082418}. 
Before computations of \eqref{eq:2308252}, let us present some preliminaries. 

It will be useful to decompose $H$ into the radial and spherical parts
with respect to \eqref{eq:2308243} and \eqref{eq:23030414}.

\begin{lemma}\label{lem:230304}
Let $\lambda\ge0$, and let $A$ and $L$ be the operators defined in
(\ref{eq:2308243}) and (\ref{eq:23030414}), respectively.  
Then one has a decomposition
\begin{align*}
H-\lambda=
\tfrac12A a^2A
+\tfrac12 L
-\tfrac12a^2
+q_1
\ \ \text{on }\mathbb R^d\setminus\{0\}
,
\end{align*}
where $q_1\in L^\infty(\mathbb R^d)$ satisfies that there exists $C>0$ such that for any $x\in\mathbb R^d$
\[
|q_1(x)|\le C\langle x\rangle^{-1-\nu'/2}.\]
\end{lemma}

\begin{proof}
By definitions \eqref{eq:2308243} and \eqref{eq:23030414} of $p_r$ and $\ell$, respectively, we first rewrite 
\begin{align}
H-\lambda=
\tfrac12 p_r^*p_r^{\phantom{*}}+\tfrac12 L-\tfrac12a^2+q
.
\label{eq:230825251}
\end{align}
Noting 
\begin{align}
\Delta r=(d-1)r^{-1}, 
\label{eq:230825344}
\end{align}
we can have 
\begin{align}
\begin{split}
A
&
= a^{-1}p_r
-\tfrac{\mathrm i}2\nabla \cdot( a^{-1}\nabla r)
=
 a^{-1}p_r
-\tfrac{\mathrm i}2 \bigl( a^{-3}(\partial_rV)+(d-1)r^{-1}a^{-1}\bigr)
,
\end{split}
\label{eq:23030413}
\end{align}
so that 
\begin{align}
p_r
=
 aA
+\tfrac{\mathrm i}2\bigl( a^{-2}(\partial_rV)
+(d-1)r^{-1}\bigr)
.
\label{eq:230825250}
\end{align}
Then by substituting \eqref{eq:230825250} into \eqref{eq:230825251}, we can proceed as 
\begin{align*}
H-\lambda
&
=
\tfrac12 \bigl[A a-\tfrac{\mathrm i}2\bigl( a^{-2}(\partial_rV)+(d-1)r^{-1}\bigr)\bigr]
\\&\phantom{{}={}}{}
\cdot\bigl[aA+\tfrac{\mathrm i}2\bigl( a^{-2}(\partial_rV)+(d-1)r^{-1}\bigr)\bigr]
+\tfrac12 L-\tfrac12a^2+q
\\&
=
\tfrac12A a^2A
+\tfrac12 L-\tfrac12a^2+q
+\tfrac14 a^{-1}\bigl[\partial_r \bigl( a^{-1}(\partial_rV)+(d-1)r^{-1}a\bigr)\bigr]
\\&\phantom{{}={}}{}
+\tfrac18\bigl( a^{-2}(\partial_rV)+(d-1)r^{-1}\bigr)^2
.
\end{align*}
Hence, we are done. 
\end{proof}

\begin{lemma}\label{lem:230925}
There exist $c,C>0$ such that for any $(\lambda,x)\in [0,\infty)\times\mathbb R^d$
\begin{align*}
cr(x) a(\lambda,x)\le S(\lambda,x)\le Cr(x) a(\lambda,x).
\end{align*}
\end{lemma}
\begin{proof}
By Assumption~\ref{cond:230806},
we have 
\begin{align}
c_1\bigl(\max\{\lambda,0\}+\langle x\rangle^{-\nu}\bigr)^{1/2}\le a(\lambda,x)
\le C_1\bigl(\max\{\lambda,0\}+\langle x\rangle^{-\nu}\bigr)^{1/2}.
\label{eq:231113}
\end{align}
Hence, the asserted bound for $S$ from below is obvious. 
As for the one from above, we can compute it by using \eqref{eq:231113} as 
\begin{align*}
S(\lambda,x)
&
\le C_1\int_0^r \bigl(\max\{\lambda^{1/2},0\}+\langle s\rangle^{-\nu/2}\bigr)\,\mathrm ds
\\&
\le C_2r\bigl(\max\{\lambda^{1/2},0\}+\langle r\rangle^{-\nu/2}\bigr)
\\&
\le 2C_2r\bigl(\max\{\lambda,0\}+\langle r\rangle^{-\nu}\bigr)^{1/2}
\\&
\le C_3ra(\lambda,x)
.
\end{align*}
Thus, the assertion follows. 
\end{proof}

For simplicity of notation, we denote  
\[
\theta_0=1+S/R. 
\]
We also denote by primes the derivatives of $\chi_{m,n}$ in $\tau$ and of $\theta,\theta_0$ in $S$,
such as  
\begin{align*}
\mathrm i[A,\theta]&=\theta'=\alpha +\beta \theta_0^{-1-\delta},\\ 
\mathrm i[A,\theta']&=\theta''=-\beta R^{-1}(1+\delta)\theta_0^{-2-\delta},\\ 
\mathrm i[A,\Theta]&=\chi_{m,n}'a^{-2}\mathrm e^{2\theta}+2\theta'\Theta
.
\end{align*}
We shall repeated use the following estimates without a reference. 

\begin{lemma}
Fix any $\lambda\ge 0$ and $\delta>0$. 
Then there exists $c>0$ such that uniformly in $R\ge 1$ 
\begin{align*}
c\langle S\rangle^{-1}\le \theta_0^{-1}\le 1.
\end{align*}
In addition, for any $k=2,3,\ldots$ there exists $C>0$ such that uniformly in $\beta\ge 0$ and $R\ge 1$ 
\begin{align*}
0\le (-1)^{k+1}\theta^{(k)}\le C\beta\langle S\rangle^{1-k}\theta_0^{-1-\delta}
.
\end{align*}
\end{lemma}
\begin{proof}
These estimates are trivial. We omit the proof. 
\end{proof}

Now we compute and bound the distorted commutator \eqref{eq:2308252}. 
The following lemma is a key for the proof of Proposition~\ref{prop:23082418}.

\begin{lemma}\label{lem:2308271}
Fix any $\lambda\ge 0$, $\alpha_0\ge 0$ and $\delta\in (0,(\nu'-\nu)/(2-\nu))$,
and let $\beta>0$ be sufficiently small. 
Then there exist $c,C>0$, $N_0\in\mathbb N_0$ and $R_0\ge 1$ such that 
uniformly in $\alpha\in [0,\alpha_0]$, $n\ge m\ge N_0$ and $R\ge R_0$
\begin{align*}
\mathop{\mathrm{Im}}(A\Theta (H-\lambda))
&\ge c r^{-1}a \theta_0^{-\delta} \Theta 
-C(\chi_{m-1,m+1}+\chi_{n-1,n+1})\tau^{-1}\mathrm e^{2\theta}
\\&\phantom{{}={}}{}
+\mathop{\mathrm{Re}}(f (H-\lambda))
,
\end{align*}
where $f$ is a multiplication operator by some function $f$, 
not a function of $H-\lambda$, that satisfies $\operatorname{supp}f\subset\operatorname{supp}\chi_{m,n}$ and
$|f|\le C_{m,n}e^{2\theta}$.
\end{lemma}

\begin{proof}
Fix $\lambda$, $\alpha_0$ and $\delta$ as in the assertion. 
We will for the moment discuss uniform estimates in 
$\beta\in (0,1]$, $\alpha\in [0,\alpha_0]$, $n\ge m\ge 0$ and $R\ge 1$,
and at the last step choose $\beta$, $N_0$ and $R_0$ so that the assertion holds. 
By Lemma~\ref{lem:230304}, we have 
\begin{align}
\begin{split}
\mathop{\mathrm{Im}}(A\Theta (H-\lambda))
&=
\tfrac12\mathop{\mathrm{Im}}(A\Theta A a^2A)
+\tfrac12 \mathop{\mathrm{Im}}(A\Theta L)
-\tfrac12\mathop{\mathrm{Im}}(A\Theta a^2)
\\&\phantom{{}={}}{}
+\mathop{\mathrm{Im}}(A\Theta q_1)
,
\end{split}
\label{eq:230303}
\end{align}
and in the following we further compute each term on the right-hand side of \eqref{eq:230303}. 
There appear many terms that will turn out to be negligible at last. 
In short, we shall gather them and write simply   
\begin{align*}
Q=
r^{-1-\nu'/2}\Theta
+Ar^{-1-\nu'/2}\Theta A
+\bigl(|\chi_{m,n}'|+|\chi_{m,n}''|a^{-2}\bigr)\mathrm e^{2\theta}
+p\cdot |\chi_{m,n}'|a^{-2}\mathrm e^{2\theta}p.
\end{align*}
In particular, once a derivative hits on $\chi_{m,n}$ in $\Theta$, the corresponding term is absorbed into $Q$. 
The term $Q$ will be computed and bounded later on.  

Now the first term on the right-hand side of \eqref{eq:230303} is rewritten and bounded by 
using \eqref{eq:23030413} and the Cauchy--Schwarz inequality as 
\begin{align}
\begin{split}
\tfrac12\mathop{\mathrm{Im}}(A\Theta A a^2A)
&
=
\tfrac14A a(\partial_r\Theta) A
+\tfrac12A  a^{-1}(\partial_rV)\Theta A
\\&
\ge 
\tfrac12A a^2\theta'\Theta A
+\tfrac12A  a^{-1}(\partial_rV)\Theta A
-C_1Q
.
\end{split}
\label{eq:23092320}
\end{align}
Here and below, $c_*,C_*>0$ are uniform in $\beta\in (0,1]$, $\alpha\in [0,\alpha_0]$, $n\ge m\ge 0$ and $R\ge 1$. 
We use the adjoint of \eqref{eq:23030413}, \eqref{eq:23030414} and \eqref{eq:230825344} 
to rewrite the second term of \eqref{eq:230303} as 
\begin{align}
\begin{split}
\tfrac12 \mathop{\mathrm{Im}}(A\Theta L)
&=
\tfrac12 \mathop{\mathrm{Im}}\bigl(p_r^* a^{-1}\Theta p_i\ell_{ij}p_j\bigr)
+\tfrac14 \bigl( a^{-3}(\partial_rV)+(d-1)r^{-1}a^{-1}\bigr)\Theta L
\\&
=
\tfrac12 \mathop{\mathrm{Im}}\bigl(p_ip_k(\nabla_kr) a^{-1}\Theta \ell_{ij}p_j\bigr)
+\tfrac12 \mathop{\mathrm{Re}}\bigl(p_k(\nabla^2r)_{ik} a^{-1}\Theta \ell_{ij}p_j\bigr)
\\&\phantom{{}={}}{}
+\tfrac14 \bigl( a^{-3}(\partial_rV)+(d-1)r^{-1}a^{-1}\bigr)\Theta L
\\&
\ge 
-\tfrac12 \theta'\Theta L
+\tfrac12 r^{-1} a^{-1}\Theta L
-C_2Q
.
\end{split}
\label{eq:23092321}
\end{align}
Here we have also used that 
\begin{align}
(\nabla_ir)\ell_{ij}=0
,\quad 
(\nabla_kr)(\partial_k\ell_{ij})
=
-r^{-1}x_k(\partial_kr^{-2}x_ix_j)
=0
.
\label{eq:2308273}
\end{align}
The third and fourth terms of \eqref{eq:230303} are computed and bounded by the Cauchy--Schwarz inequality as 
\begin{align*}
-\tfrac12\mathop{\mathrm{Im}}(A\Theta a^2)
+\mathop{\mathrm{Im}}(A\Theta q_1)
&\ge 
\tfrac12a^2\theta'\Theta
-\tfrac12 a^{-1}(\partial_rV)\Theta 
-C_3Q
.
\end{align*}
Thus, we have \eqref{eq:230303} bounded as 
\begin{align}
\begin{split}
\mathop{\mathrm{Im}}(A\Theta (H-\lambda))
&\ge 
\tfrac12A a^2\theta'\Theta A
+\tfrac12A  a^{-1}(\partial_rV)\Theta A
-\tfrac12\theta'\Theta L
+\tfrac12 r^{-1} a^{-1}\Theta L
\\&\phantom{{}={}}{}
-\tfrac12a^2\theta'\Theta 
-\tfrac12 a^{-1}(\partial_rV)\Theta 
-C_4Q
.
\end{split}
\label{eq:23030321}
\end{align}

We continue to compute the right-hand side of \eqref{eq:23030321}. 
Using Lemma~\ref{lem:230304} and $S^{-1}a\le C_5r^{-1}$, we combine the third and fifth terms of \eqref{eq:23030321} as 
\begin{align*}
-\tfrac12\theta' \Theta L
+\tfrac12a^2\theta' \Theta 
&
=
\tfrac12\mathop{\mathrm{Re}}\bigl(\theta' \Theta A a^2A\bigr)
+q_1\theta' \Theta
-\mathop{\mathrm{Re}}(\theta' \Theta(H-\lambda))
\\&
=
\tfrac12Aa^2\theta'\Theta  A
-\tfrac14a^{-1}\bigl(\partial_ra(\partial_r\theta' \Theta )\bigr)
+q_1\theta' \Theta
-\mathop{\mathrm{Re}}(\theta' \Theta(H-\lambda))
\\&
\ge 
\tfrac12Aa^2\theta'\Theta  A
+a^{-1} (\partial_rV)\theta'^2 \Theta 
-a^2\theta'^3 \Theta 
-\tfrac32a^2\theta'\theta'' \Theta 
\\&\phantom{{}={}}{}
-C_6Q
-\mathop{\mathrm{Re}}(\theta' \Theta(H-\lambda))
.
\end{align*}
Similarly, by Lemma~\ref{lem:230304}, Assumption~\ref{cond:230806b} and $S^{-1}a\le C_5r^{-1}$
the second, fourth and sixth terms of \eqref{eq:23030321} are combined as 
\begin{align*}
&
\tfrac12A  a^{-1}(\partial_rV)\Theta A
+\tfrac12 r^{-1} a^{-1}\Theta L
-\tfrac12 a^{-1}(\partial_rV)\Theta 
\\&
=
\tfrac12 r^{-1} a^{-3}(a^2-r\partial_rV)\Theta L
+\tfrac14a^{-1}\bigl(\partial_ra\bigl(\partial_ra^{-3}(\partial_rV)\Theta\bigr)\bigr)
\\&\phantom{{}={}}{}
- a^{-3}(\partial_rV) q_1\Theta
+\mathop{\mathrm{Re}}\bigl( a^{-3}(\partial_rV)\Theta (H-\lambda)\bigr)
\\&
\ge 
c_1r^{-1} a^{-1}\Theta L
+a^{-1}(\partial_rV)\theta'^2\Theta  
-C_7Q
+\mathop{\mathrm{Re}}\bigl( a^{-3}(\partial_rV)\Theta (H-\lambda)\bigr)
.
\end{align*}
Thus, we obtain 
\begin{align}
\begin{split}
\mathop{\mathrm{Im}}(A\Theta (H-\lambda))
&\ge 
Aa^2\theta'\Theta  A
+c_1r^{-1} a^{-1}\Theta L
-a^2\theta'^3 \Theta 
-\tfrac32a^2\theta'\theta'' \Theta 
\\&\phantom{{}={}}{}
+2a^{-1}(\partial_rV)\theta'^2\Theta  
-C_8Q
+\mathop{\mathrm{Re}}( f_1(H-\lambda))
.
\end{split}
\label{eq:23030322}
\end{align}
Here and below, $f_*$ satisfy the same conditions as $f$ in the assertion. 

We can see the third term on the right-hand side of  \eqref{eq:23030322} is the worst negative contribution. 
To remove it, we further exploit the first and second terms of \eqref{eq:23030322} as follows.  
Let us split them as 
\begin{align}
\begin{split}
Aa^2\theta'\Theta  A
+c_1r^{-1} a^{-1}\Theta L
&
\ge 
A\bigl(a^2\theta'-c_1r^{-1} a\theta_0^{-\delta}\bigr)\Theta A
\\&\phantom{{}={}}{}
+c_1Ar^{-1}a\theta_0^{-\delta}\Theta A
+c_1r^{-1}a^{-1}\theta_0^{-\delta} \Theta L
.
\end{split}
\label{eq:23030415}
\end{align}
Then the first term of \eqref{eq:23030415} is bounded by using  $S^{-1}a\le C_5r^{-1}$ as 
\begin{align*}
A\bigl(a^2\theta'-c_1r^{-1} a\theta_0^{-\delta}\bigr)\Theta A
&=
(A+\mathrm i\theta')\bigl(a^2\theta'-c_1r^{-1}a\theta_0^{-\delta} \bigr)\Theta (A-\mathrm i\theta')
\\&\phantom{{}={}}{}
+ a^{-1}\bigl(\partial_r(a^2\theta'-c_1r^{-1} a\theta_0^{-\delta})\theta'\Theta\bigr) 
\\&\phantom{{}={}}{}
-\bigl(a^2\theta'-c_1r^{-1} a\theta_0^{-\delta}\bigr)\theta'^2\Theta 
\\&\ge 
(A+\mathrm i\theta')\bigl(\beta a^2\theta_0^{-1-\delta}-c_1r^{-1}a\theta_0^{-\delta} \bigr)\Theta (A-\mathrm i\theta')
+ a^2\theta'^3\Theta 
\\&\phantom{{}={}}{}
+ 2a^2\theta'\theta''\Theta
-2a^{-1} (\partial_rV)\theta'^2\Theta
-c_1r^{-1} a\theta_0^{-\delta}\theta'^2\Theta 
-C_9Q
.
\end{align*}
On the other hand, by Lemma~\ref{lem:230304} and  $S^{-1}a\le C_5r^{-1}$ 
the second and third terms of \eqref{eq:23030415}
are combined as 
\begin{align*}
&
c_1Ar^{-1}a\theta_0^{-\delta} \Theta A
+c_1r^{-1}a^{-1}\theta_0^{-\delta} \Theta L
\\&=
\tfrac{c_1}2a^{-1}\bigl(\partial_r a\bigl(\partial_r r^{-1} a^{-1}\theta_0^{-\delta}\Theta\bigr)\bigr)
+c_1 r^{-1}a\theta_0^{-\delta} \Theta 
\\&\phantom{{}={}}{}
-2c_1 r^{-1}a^{-1}\theta_0^{-\delta} q_1\Theta 
+2c_1\mathop{\mathrm{Re}}\bigl(r^{-1}a^{-1}\theta_0^{-\delta} \Theta (H-\lambda)\bigr)
\\&\ge 
c_1 r^{-1}a\theta_0^{-\delta} \Theta 
+2c_1r^{-1}a\theta_0^{-\delta} \theta'^2\Theta  
-C_{10}Q
+2c_1\mathop{\mathrm{Re}}\bigl(r^{-1}\theta_0^{-\delta} a^{-1}\Theta (H-\lambda)\bigr)
.
\end{align*}
At this stage, we also bound $Q$, similarly to so far, as 
\begin{align*}
Q
&
\le 
C_{11}r^{-1-\nu'/2}\Theta
+C_{11}(\chi_{m-1,m+1}+\chi_{n-1,n+1})\tau^{-1}\mathrm e^{2\theta}
\\&\phantom{{}={}}{}
+2\mathop{\mathrm{Re}}\bigl(r^{-1-\nu'/2}a^{-2}\Theta (H-\lambda)\bigr)
+2\mathop{\mathrm{Re}}\bigl(|\chi_{m,n}'|a^{-2}\mathrm e^{2\theta}(H-\lambda)\bigr)
.
\end{align*}
Therefore, we obtain 
\begin{align}
\begin{split}
\mathop{\mathrm{Im}}(A\Theta (H-\lambda))
&\ge 
(A+\mathrm i\theta')\bigl(\beta a^2\theta_0^{-1-\delta}-c_1r^{-1}a\theta_0^{-\delta} \bigr)\Theta (A-\mathrm i\theta')
\\&\phantom{{}={}}{}
+c_1 r^{-1}a\theta_0^{-\delta} \Theta 
+c_1r^{-1} a\theta_0^{-\delta}\theta'^2\Theta 
+\tfrac12a^2\theta'\theta''\Theta
\\&\phantom{{}={}}{}
-C_{12}r^{-1-\nu'/2}\Theta
-C_{12}(\chi_{m-1,m+1}+\chi_{n-1,n+1})\tau^{-1}\mathrm e^{2\theta}
\\&\phantom{{}={}}{}
+\mathop{\mathrm{Re}}\bigl(f_2(H-\lambda)\bigr)
\\&\ge 
(A+\mathrm i\theta')
\bigl(\beta a r-c_1R^{-1}S-c_1 \bigr)r^{-1}a\theta_0^{-1-\delta}\Theta (A-\mathrm i\theta')
\\&\phantom{{}={}}{}
+c_2 r^{-1}a\theta_0^{-\delta} \Theta 
+(c_2-C_{13}\beta) r^{-1}a\theta_0^{-\delta} \Theta 
\\&\phantom{{}={}}{}
+\bigl(c_2r^{-1-\nu/2}S^{-\delta}-C_{12}r^{-1-\nu'/2}\bigr) \Theta 
\\&\phantom{{}={}}{}
-C_{12}(\chi_{m-1,m+1}+\chi_{n-1,n+1})\tau^{-1}\mathrm e^{2\theta}
+\mathop{\mathrm{Re}}\bigl(f_2(H-\lambda)\bigr)
.
\end{split}
\label{eq:230827}
\end{align}

Now we choose and fix $\beta>0$ small enough that the third term on the right-hand side of \eqref{eq:230827}
is nonnegative. Then we can choose $N_0\in\mathbb N$ and $R_0\ge 1$ large 
so that the first and fourth terms are nonnegative for any $n\ge m\ge N_0$ and $R\ge R_0$. Hence, we are done. 
\end{proof}

\begin{proof}[Proof of Proposition~\ref{prop:23082418}]
Let $\phi\in \mathcal B^*_0(\lambda)$ and $\lambda\ge 0$ satisfy the assumption of the assertion, 
and set
\begin{align*}
\alpha_0=\sup\bigl\{\alpha\ge 0\,\big|\,\mathrm e^{\alpha S}\phi\in \mathcal B_0^*(\lambda)\bigr\}.
\end{align*}
Assume $\alpha_0<\infty$, and we deduce a contradiction.
Fix $\delta$, $\beta$, $N_0$ and $R_0$ as in the assertion of Lemma~\ref{lem:2308271},
and take any $\alpha\in \{0\}\cup [0,\alpha_0)$ such that $\alpha+\beta>\alpha_0$. 
With such parameters, evaluate the inequality from Lemma~\ref{lem:2308271} for the state $\chi_{m-2,n+2}\phi$, 
and then, we obtain for any $n\ge m\ge N_0$ and $R\ge R_0$
\begin{align}
\begin{split}
\bigl\|(r^{-1}a\theta_0^{-\delta}\Theta)^{1/2}\phi\bigr\|^2
&
\le 
C_m\bigl\|\chi_{m-1,m+1}^{1/2}\phi\bigr\|^2
+C_R 2^{-n/2}\bigl\|\chi_{n-1,n+1}^{1/2}\mathrm e^{\alpha S}\phi\bigr\|^2
.
\end{split}
\label{eq:11.7.16.3.22a}
\end{align}
The second term on the right-hand side of (\ref{eq:11.7.16.3.22a})
  vanishes in the limit $n\to\infty$ by the assumption, 
and hence, by Lebesgue's monotone  convergence
theorem
\begin{align}
\bigl\|(r^{-1}a\theta_0^{-\delta}\Theta)^{1/2}\phi\bigr\|^2
 &\le 
C_m\bigl\|\chi_{m-1,m+1}^{1/2}\phi\bigr\|^2.
\label{eq:11.7.16.3.43a}
\end{align}
Next we let $R \to\infty$ in \eqref{eq:11.7.16.3.43a}
    invoking again Lebesgue's monotone  convergence
theorem,
and then, it follows that 
 $$\bar\chi_m^{1/2} r^{-1/2}a^{1/2}\mathrm e^{(\alpha+\beta)S}\phi
\in L^2(\mathbb R^d).$$ 
 This implies $\mathrm e^{\kappa S}\phi\in 
\mathcal B_0^*(\lambda)$ for any $\kappa\in(0,\alpha+\beta)$, 
contradicting $\alpha+\beta>\alpha_0$. 
Thus, we are done. 
\end{proof}

\subsection{Absence of super-exponentially decaying eigenfunction}\label{subsec:2308251} 

Here we prove Propositions~\ref{prop:23082419}. 
The proof is very similar to that of Propositions~\ref{prop:23082418}, 
but we focus on different parameters. 
In fact, we let $\beta=0$, so that 
\[
\theta=2\alpha S,\quad 
\Theta=\chi_{m,n}\mathrm e^{\alpha S}, 
\]
and we deduce uniform estimates in $\alpha\ge 0$ as follows. Note here $\delta$ and $R$ are irrelevant.

\begin{lemma}\label{lem:23082716}
Fix any $\lambda\ge 0$, and let $\beta=0$. 
Then there exist $c,C>0$ and $N_0\in\mathbb N_0$ such that 
uniformly in $\alpha\ge 1$ and $n\ge m\ge N_0$ 
\begin{align*}
\mathop{\mathrm{Im}}(A\Theta (H-\lambda))
&\ge c \alpha^2r^{-1}a \Theta 
-C\alpha^2(\chi_{m-1,m+1}+\chi_{n-1,n+1})\tau^{-1}\mathrm e^{2\theta}
\\&\phantom{{}={}}{}
+\mathop{\mathrm{Re}}(f (H-\lambda))
,
\end{align*}
where $f$ is a multiplication operator by some function $f$, 
not a function of $H-\lambda$, that satisfies 
$\operatorname{supp}f\subset\operatorname{supp}\chi_{m,n}$ and
$|f|\le C\alpha e^{2\theta}$
uniformly in $\alpha\ge1$ and $n\ge m\ge N_0$.
\end{lemma}

\begin{proof}
We repeat computations similar to the proof of Lemma~\ref{lem:2308271},
however, focusing on different parameters.  
Fix $\lambda$ and $\beta$ as in the assertion.
We will again for the moment discuss uniform estimates in 
$\alpha\ge 1$ and $n\ge m\ge 0$, and then finally fix appropriate $N_0$ so that the assertion holds. 
By Lemma~\ref{lem:230304}, we write  
\begin{align}
\begin{split}
\mathop{\mathrm{Im}}(A\Theta (H-\lambda))
&=
\tfrac12\mathop{\mathrm{Im}}(A\Theta A a^2A)
+\tfrac12 \mathop{\mathrm{Im}}(A\Theta L)
-\tfrac12\mathop{\mathrm{Im}}(A\Theta a^2)
\\&\phantom{{}={}}{}
+\mathop{\mathrm{Im}}(A\Theta q_1)
,
\end{split}
\label{eq:230303b}
\end{align}
and further compute the right-hand side below. 
The terms to be negligible are gathered and denoted by 
\begin{align*}
Q
&=
\alpha r^{-1-\nu'/2}\Theta
+\alpha^{-1}Ar^{-1-\nu'/2}\Theta A
\\&\phantom{{}={}}{}
+\bigl(\alpha^2|\chi_{m,n}'|+\alpha|\chi_{m,n}''|a^{-2}\bigr)\mathrm e^{2\theta}
+p\cdot |\chi_{m,n}'|a^{-2}\mathrm e^{2\theta}p
,
\end{align*}
and it will be estimated later on. 

The first term on the right-hand side of \eqref{eq:230303b} is rewritten and bounded by 
using \eqref{eq:23030413} and the Cauchy--Schwarz inequality as 
\begin{align*}
\tfrac12\mathop{\mathrm{Im}}(A\Theta A a^2A)
&
=
\tfrac14A a(\partial_r\Theta) A
+\tfrac12A  a^{-1}(\partial_rV)\Theta A
\\&
\ge 
\tfrac{\alpha}2A a^2\Theta A
+\tfrac12A  a^{-1}(\partial_rV)\Theta A
-C_1Q
.
\end{align*}
Here and below, $c_*,C_*>0$ are uniform in $\alpha\ge 1$ and $n\ge m\ge 0$. 
By the adjoint of \eqref{eq:23030413}, \eqref{eq:23030414}, \eqref{eq:230825344} and \eqref{eq:2308273},
we compute the second term of \eqref{eq:230303b} as 
\begin{align*}
\tfrac12 \mathop{\mathrm{Im}}(A\Theta L)
&=
\tfrac12 \mathop{\mathrm{Im}}\bigl(p_r^* a^{-1}\Theta p_i\ell_{ij}p_j\bigr)
+\tfrac14 \bigl( a^{-3}(\partial_rV)+(d-1)r^{-1}a^{-1}\bigr)\Theta L
\\&
=
\tfrac12 \mathop{\mathrm{Im}}\bigl(p_ip_k(\nabla_kr) a^{-1}\Theta \ell_{ij}p_j\bigr)
+\tfrac12 \mathop{\mathrm{Re}}\bigl(p_k(\nabla^2r)_{ik} a^{-1}\Theta \ell_{ij}p_j\bigr)
\\&\phantom{{}={}}{}
+\tfrac14 \bigl( a^{-3}(\partial_rV)+(d-1)r^{-1}a^{-1}\bigr)\Theta L
\\&
\ge 
-\tfrac{\alpha}2 \Theta L
+\tfrac12 r^{-1} a^{-1}\Theta L
-C_2Q
.
\end{align*}
The third and fourth terms of \eqref{eq:230303b} are computed and bounded by the Cauchy--Schwarz inequality as 
\begin{align*}
-\tfrac12\mathop{\mathrm{Im}}(A\Theta a^2)
+\mathop{\mathrm{Im}}(A\Theta q_1)
&\ge 
\tfrac{\alpha}2a^2\Theta
-\tfrac12 a^{-1}(\partial_rV)\Theta 
-C_3Q
.
\end{align*}
Thus, \eqref{eq:230303b} is bounded as 
\begin{align}
\begin{split}
\mathop{\mathrm{Im}}(A\Theta (H-\lambda))
&\ge 
\tfrac{\alpha}2A a^2\Theta A
+\tfrac12A  a^{-1}(\partial_rV)\Theta A
-\tfrac{\alpha}2\Theta L
+\tfrac12 r^{-1} a^{-1}\Theta L
\\&\phantom{{}={}}{}
-\tfrac{\alpha}2a^2\Theta 
-\tfrac12 a^{-1}(\partial_rV)\Theta 
-C_4Q
.
\end{split}
\label{eq:23030321b}
\end{align}

We continue to compute the right-hand side of \eqref{eq:23030321b}. 
Using Lemma~\ref{lem:230304}, we combine the third and fifth terms of \eqref{eq:23030321b} as 
\begin{align*}
-\tfrac{\alpha}2 \Theta L
+\tfrac{\alpha}2a^2 \Theta 
&
=
\tfrac{\alpha}2\mathop{\mathrm{Re}}\bigl(\Theta A a^2A\bigr)
+\alpha q_1 \Theta
-\alpha\mathop{\mathrm{Re}}(\Theta(H-\lambda))
\\&
=
\tfrac{\alpha}2Aa^2\Theta  A
-\tfrac{\alpha}4a^{-1}\bigl(\partial_ra(\partial_r \Theta )\bigr)
+\alpha q_1 \Theta
-\alpha\mathop{\mathrm{Re}}( \Theta(H-\lambda))
\\&
\ge 
\tfrac{\alpha}2Aa^2\Theta  A
+\alpha^2a^{-1} (\partial_rV)   \Theta 
-\alpha^3 a^2   \Theta 
-C_5Q
-\alpha\mathop{\mathrm{Re}}( \Theta(H-\lambda))
.
\end{align*}
Similarly, by Lemma~\ref{lem:230304} and Assumption~\ref{cond:230806b}
the second, fourth and sixth terms of \eqref{eq:23030321b} are combined as 
\begin{align*}
&
\tfrac12A  a^{-1}(\partial_rV)\Theta A
+\tfrac12 r^{-1} a^{-1}\Theta L
-\tfrac12 a^{-1}(\partial_rV)\Theta 
\\&
=
\tfrac12 r^{-1} a^{-3}(a^2-r\partial_rV)\Theta L
+\tfrac14a^{-1}\bigl(\partial_ra\bigl(\partial_ra^{-3}(\partial_rV)\Theta\bigr)\bigr)
\\&\phantom{{}={}}{}
- a^{-3}(\partial_rV) q_1\Theta
+\mathop{\mathrm{Re}}\bigl( a^{-3}(\partial_rV)\Theta (H-\lambda)\bigr)
\\&
\ge 
c_1 r^{-1} a^{-1}\Theta L
+\alpha^2 a^{-1}(\partial_rV)\Theta   
-C_6Q
+\mathop{\mathrm{Re}}\bigl( a^{-3}(\partial_rV)\Theta (H-\lambda)\bigr)
.
\end{align*}
Thus, we obtain 
\begin{align}
\begin{split}
\mathop{\mathrm{Im}}(A\Theta (H-\lambda))
&\ge 
\alpha Aa^2\Theta  A
+c_1r^{-1} a^{-1}\Theta L
-\alpha^3a^2\Theta 
\\&\phantom{{}={}}{}
+2\alpha^2a^{-1}(\partial_rV)\Theta  
-C_7Q
+\mathop{\mathrm{Re}}( f_1(H-\lambda))
.
\end{split}
\label{eq:23030322b}
\end{align}
Here and below, $f_*$ satisfy the same conditions as $f$ in the assertion. 

The first to fourth terms of \eqref{eq:23030322b} are further bounded as 
\begin{align*}
&\alpha Aa^2\Theta  A
+c_1r^{-1} a^{-1}\Theta L
-\alpha^3a^2\Theta 
+2\alpha^2a^{-1}(\partial_rV)\Theta  
\\&
\ge 
\alpha (A+\mathrm i\alpha) a^2\Theta (A-\mathrm i\alpha)
+c_1r^{-1} a^{-1}\Theta L
-C_8Q
\\&\ge 
c_2 (A+\mathrm i\alpha) r^{-1}a\Theta (A-\mathrm i\alpha)
+c_2r^{-1} \Theta L
-C_8Q
\\&\ge 
c_2 A r^{-1}a\Theta A
+c_2r^{-1} \Theta L
-c_2 \alpha^2 r^{-1}a\Theta 
+c_2\alpha   r^{-2}\Theta 
-C_9Q
\\&\ge 
\tfrac{c_2}2 a^{-1}\bigl(\partial_ra \bigl(\partial_rr^{-1}a^{-1}\Theta\bigr) \bigr)
+c_2 r^{-1}a\Theta 
-2c_2 r^{-1}a^{-1}q_1\Theta 
-c_2 \alpha^2 r^{-1}a\Theta 
\\&\phantom{{}={}}{}
+c_2\alpha   r^{-2}\Theta 
-C_9Q
+2c_2 \mathop{\mathrm{Re}}\bigl( r^{-1}a^{-1}\Theta (H-\lambda)\bigr)
\\&\ge 
c_2 (\alpha^2+1) r^{-1}a\Theta 
-C_{10}Q
+2c_2 \mathop{\mathrm{Re}}\bigl( r^{-1}a^{-1}\Theta (H-\lambda)\bigr)
.
\end{align*}
We also bound $Q$, similarly to so far, as 
\begin{align*}
Q
&
\le 
C_{11}\alpha r^{-1-\nu'/2}\Theta
+C_{11}\alpha^2(\chi_{m-1,m+1}+\chi_{n-1,n+1})\tau^{-1}\mathrm e^{2\theta}
\\&\phantom{{}={}}{}
+2\alpha^{-1}\mathop{\mathrm{Re}}\bigl(r^{-1-\nu'/2}a^{-2}\Theta (H-\lambda)\bigr)
+2\mathop{\mathrm{Re}}\bigl(|\chi_{m,n}'|a^{-2}(H-\lambda)\bigr).
.
\end{align*}
Therefore, we obtain 
\begin{align*}
\mathop{\mathrm{Im}}(A\Theta (H-\lambda))
&\ge 
c_2 (\alpha^2+1) r^{-1}a\Theta 
-C_{12}\alpha r^{-1-\nu'/2}\Theta
\\&\phantom{{}={}}{}
-C_{12}\alpha^2(\chi_{m-1,m+1}+\chi_{n-1,n+1})\tau^{-1}\mathrm e^{2\theta}
+\mathop{\mathrm{Re}}( f_2(H-\lambda))
.
\end{align*}
Hence, by letting $N_0\in\mathbb N_0$ be large enough, we obtain the assertion. 
\end{proof}

\begin{proof}[Proof of Propositions~\ref{prop:23082419}]
Let $\phi \in \mathcal B_0^*(\lambda)$ and $\lambda\ge 0$
satisfy the assumption of the assertion.
Choose $N_0 \ge 0$ as in Lemma~\ref{lem:23082716},
and we evaluate the inequality from Lemma~\ref{lem:23082716} for the state $\chi_{m-2,n+2}\phi$.
Then $\alpha\ge 1$ and $n\ge m\ge N_0$
\begin{align}\label{eq:evaluate-in-the-state}
\begin{split}
\bigl\| (r^{-1}a\chi_{m,n})^{1/2}\mathrm e^{\alpha S} \phi \bigr\|^2 
&\le 
C_1\bigl\| \chi_{m-1,m+1}^{1/2}\mathrm e^{\alpha S}\phi\bigr\|^2 
+ C_12^{-n/2}\bigl\| \chi_{n-1,n+1}^{1/2}\mathrm e^{\alpha S}\phi \bigr\|^2.
\end{split}
\end{align}
Since $\mathrm e^{\alpha S}\phi \in \mathcal B_0^*(\lambda)$ for any $\alpha\ge 0$, 
the second term on the right-hand side of \eqref{eq:evaluate-in-the-state} 
vanishes in the limit $n \to \infty$.
Hence, by the Lebesgue monotone convergence theorem we obtain
\begin{equation}\label{eq:n-infty}
\bigl\| (r^{-1}a\bar \chi_m)^{1/2}\mathrm e^{\alpha S} \phi \bigr\|^2 
\le
 C_1\bigl\| \chi_{m-1,m+1}\mathrm e^{\alpha S}\phi \bigr\|^2,
\end{equation}
Now assume $\bar\chi_{m+2}\phi \neq 0$,
and then, we can deduce a contradiction from \eqref{eq:n-infty} as $\alpha\to\infty$. 
Thus, $\bar\chi_{m+2}\phi =0$.
By invoking the unique continuation property for the second-order elliptic operator
$H$, we conclude that $\phi \equiv 0$. We refer to~\cite{W}
for the unique continuation property when $d \ge 2$. The case $d = 1$ follows
from the uniqueness of solutions to ordinary differential equations.
\end{proof}

\begin{proof}[Proof of Theorem~\ref{thm:230607}]
The assertion is obvious by Propositions~\ref{prop:23082418} and \ref{prop:23082419}. 
\end{proof}

\section{LAP bounds}\label{2411244}

\subsection{Main proposition}

In this section, we prove the LAP bounds, or Theorem~\ref{thm:2308271915}. 
The proof depends on a commutator argument as in the previous section, 
but we use a different weight function of the form 
\begin{align*}
\Theta=\Theta_{R}^\delta=\bar\chi_0\theta,
\end{align*}
where $\bar\chi_0$ is from \eqref{eq:11.7.11.5.14} with $n=0$, and 
\begin{align}
\theta
=\int_0^{\tau/R}(1+s)^{-1-\delta}\,\mathrm ds
=\delta^{-1}\bigl[1-(1+\tau/R)^{-\delta}\bigr]
;\quad \delta>0,\ \ R\ge1.
\label{eq:15.2.15.5.8}
\end{align}
Note we put $\bar\chi_0$ to remove a singularity of $A$ at the origin. 
In this section, we denote derivatives of functions in $\tau$ by primes, such as 
\begin{align}
\theta'=R^{-1}(1+\tau/R)^{-1-\delta},\quad
\theta''=-R^{-2}(1+\delta)(1+\tau/R)^{-2-\delta}.
\label{eq:15.2.15.5.9}
\end{align}
We quote the following estimates from \cite{IS}. 

\begin{lemma}\label{lem:theta-inequality}
For any $\delta>0$, there exist $c, C, C_k>0,\ k=2, 3, \ldots$, such that for any $k=2,3,\ldots$ 
and uniformly in $R\ge 1$
\begin{align*}
&
\min\{c,c\tau/R\} \le \theta \le \min\{C, \tau/R\},
\\&
c(\min\{R, \tau\})^\delta \tau^{-1-\delta}\theta \le \theta' \le \tau^{-1}\theta,
\\&
0\le(-1)^{k-1}\theta^{(k)} \le C_k\tau^{-k}\theta.
\end{align*}
\end{lemma}
\begin{proof}
We omit the proof, see, e.g., \cite[Lemma~4.2]{IS}.
\end{proof}

\begin{lemma}\label{lem:230925b}
There exist $c,C>0$ such that for any $(\lambda,x)\in \mathbb{R} \times\mathbb R^d$
\begin{align*}
c r(x) a(\lambda,x)^{-1}\le \tau(\lambda,x)\le C r(x) a(\lambda,x)^{-1}
\end{align*}
\end{lemma}
\begin{proof}
By \eqref{eq:231113}, the bound from above is easy. 
As for the bound from below, we estimate it by using \eqref{eq:231113} as 
\begin{align*}
\tau(\lambda,x)
&
\ge c_1 \int_{r/2}^r\bigl(\max\{\lambda,0\}+\langle s\rangle^{-\nu}\bigr)^{-1/2}\,\mathrm ds
\\&
\ge \tfrac{c_1}2  r\bigl(\max\{\lambda,0\}+\langle \tfrac12 r\rangle^{-\nu}\bigr)^{-1/2}
\\&
\ge c_2 ra(\lambda,x)^{-1}
.
\end{align*}
Hence, we are done. 
\end{proof}

We now present a key estimate for the proof of LAP bounds. 
It is essentially a consequence of commutator computations. 

\begin{proposition}\label{prop:23112315}
Fix any $\rho>0$, $\omega \in (0,\pi)$ and $\delta\in (0,(\nu'-\nu)/(2+\nu))$. 
Then there exist $C>0$ and $n \in\mathbb N_0$ such that 
for any $R\ge 1$, $z=\lambda\pm\mathrm i\mu\in \Gamma_{\pm}(\rho,\omega)$ and $\psi \in \mathcal B(\lambda)$
the state $\phi=R(z)\psi$ satisfies
\begin{align*}
&
 \|\theta'^{1/2}\phi\|^2 
+  \|a^{-1}\theta'^{1/2}p_r\phi\|^2 
+ \bigl\langle \phi,a^{-2}\langle \tau\rangle^{-1}\theta L\phi\bigr\rangle
\\&\le 
C\left(
\|\psi\|_{\mathcal B(\lambda)}^2
+ \|\phi\|_{\mathcal B^*(\lambda)} \|\psi\|_{\mathcal B(\lambda)} 
+ \|a^{-1}p_r\phi\|_{\mathcal B^*(\lambda)} \|\psi\|_{\mathcal B(\lambda)} 
+ R^{-1}\|\chi_n^{1/2}\phi\|^2 \right).
\end{align*}
\end{proposition}
\begin{proof}
Fix $\rho$, $\omega$ and $\delta$ as in the assertion. 
Clearly it suffices to show there exist $n\in\mathbb N_0$ and $c_1,C_1>0$ such that 
\begin{equation}
\begin{split}
\mathop{\mathrm{Im}}( A\Theta(H-z)) 
&\ge 
c_1 \theta' +c_1 p_r^*a^{-2}\theta'p_r +c_1a^{-2}\langle \tau\rangle^{-1}\theta L
\\&\phantom{{}={}}{}
-\C R^{-1}\chi_n 
+\mathop{\mathrm{Re}}( \gamma(H-z))
-C_1 (H-\bar z)\tau(H-z),
\end{split}
\label{eq:181218b}
\end{equation} 
uniformly in $R\ge 1$ and $z=\lambda\pm\mathrm i\mu\in \Gamma_{\pm}(\rho,\omega)$, 
where $\gamma=\gamma_{z,R}$ is a certain uniformly bounded function: 
$|\gamma| \le \C$. 
In the below, we compute and bound the quantity on the left-hand side of \eqref{eq:181218b}. 
As in the proof of Propositions~\ref{prop:23082418} and \ref{prop:23082419}, 
we gather \textit{admissible error terms} and write them for short as 
\[
Q=
r^{-1-\nu'/2}\Theta
+Ar^{-1-\nu'/2}\Theta A
+R^{-1} p_i\chi_1 p_i
. 
\]
For the moment, the following estimates are all uniform in $R\ge 1$, $z=\lambda\pm\mathrm i\mu\in\Gamma_{\pm}(\rho,\omega)$
and $n\in \mathbb N$, and we will fix $n$ only at the last step of the proof.

Similarly to \eqref{eq:230303}, we first rewrite  
\begin{align}
\begin{split}
\mathop{\mathrm{Im}}(A\Theta (H-z))
&
=
\tfrac12\mathop{\mathrm{Im}}(A\Theta Aa^2A)
+\tfrac12\mathop{\mathrm{Im}}(A\Theta L)
-\tfrac12\mathop{\mathrm{Im}}(A\Theta a^2)
\\&\phantom{{}={}}{}
+\mathop{\mathrm{Im}}(A\Theta q_1)
-\mathop{\mathrm{Im}}(A\Theta z)
.
\end{split}
\label{eq:230804}
\end{align}
Partially similarly to \eqref{eq:23092320}, the first term of \eqref{eq:230804} is computed as 
\begin{align*}
\tfrac12\mathop{\mathrm{Im}}(A\Theta Aa^2A)
&=
\tfrac14Aa(\partial_r\Theta) A
-\tfrac14Aa^{-1}(\partial_ra^2)\Theta A
\\&
=
\tfrac14 A\Theta' A
+\tfrac12A a^{-1}(\partial_rV)\Theta A
.
\end{align*}
We can also use a part of \eqref{eq:23092321} to compute 
the second term of \eqref{eq:230804} as 
\begin{align*}
\tfrac12\mathop{\mathrm{Im}}(A\Theta L)
&=
\tfrac12 \mathop{\mathrm{Im}}\bigl(p_ip_k(\nabla_kr) a^{-1}\Theta \ell_{ij}p_j\bigr)
+\tfrac12 \mathop{\mathrm{Re}}\bigl(p_k(\nabla^2r)_{ik} a^{-1}\Theta \ell_{ij}p_j\bigr)
\\&\phantom{{}={}}{}
+\tfrac14 \bigl( a^{-3}(\partial_rV)+(d-1)r^{-1}a^{-1}\bigr)\Theta L
\\&
=
-\tfrac14  a^{-2}\Theta' L
+\tfrac12 r^{-1}a^{-1}\Theta L
.
\end{align*}
We can directly compute and bound the third to fifth terms of \eqref{eq:230804} by using 
the Cauchy--Schwarz inequality as 
\begin{align*}
&-\tfrac12\mathop{\mathrm{Im}}(A\Theta a^2)
+\mathop{\mathrm{Im}}(A\Theta q_1)
-\mathop{\mathrm{Im}}(A\Theta z)
\\&=
\tfrac14a^{-1}(\partial_r\Theta a^2)
+\mathop{\mathrm{Im}}(A\Theta q_1)
+\tfrac12\lambda a^{-1}(\partial_r\Theta )
\mp\mu\mathop{\mathrm{Re}}(A\Theta )
\\&
\ge 
\tfrac14 \Theta' 
-\tfrac12a^{-1}(\partial_rV)\Theta 
+\tfrac12\lambda  a^{-2}\Theta'
\mp\mu\mathop{\mathrm{Re}}(A\Theta )
-\C Q
.
\end{align*}
Thus, it follows that 
\begin{align}
\begin{split}
\mathop{\mathrm{Im}}(A\Theta (H-z))
&
\ge 
\tfrac14 A\Theta' A
+\tfrac12A a^{-1}(\partial_rV)\Theta A
-\tfrac14  a^{-2}\Theta' L
+\tfrac12 r^{-1}a^{-1}\Theta L
\\&\phantom{{}={}}{}
+\tfrac14 \Theta' 
-\tfrac12a^{-1}(\partial_rV)\Theta 
+\tfrac12\lambda  a^{-2}\Theta'
\mp\mu\mathop{\mathrm{Re}}(A\Theta )
-C_3 Q
.
\end{split}
\label{eq:23080416}
\end{align}

We continue to compute the right-hand side of \eqref{eq:23080416}. 
We can combine the third, fifth and seventh terms of \eqref{eq:23080416} as 
\begin{align}
\begin{split}
-\tfrac14 a^{-2}\Theta' L
+\tfrac14 \Theta' 
+\tfrac12\lambda  a^{-2}\Theta'
&
=
\tfrac14  \mathop{\mathrm{Re}}\bigl(a^{-2}\Theta' Aa^2A\bigr)
+\tfrac12  a^{-2} q_1\Theta'
\\&\phantom{{}={}}{}
-\tfrac12  \mathop{\mathrm{Re}}\bigl(a^{-2}\Theta' (H-z)\bigr)
\\&
\ge 
\tfrac14  A\Theta' A
-\C Q
-\D{0}R^{-1}\chi_1
\\&\phantom{{}={}}{}
-\tfrac12  \mathop{\mathrm{Re}}\bigl(a^{-2}\Theta' (H-z)\bigr)
.
\end{split}
\label{eq:2311202021}
\end{align}
The second and sixth terms of \eqref{eq:23080416} are rewritten and bounded by the 
Cauchy--Schwarz inequality as 
\begin{align}
\begin{split}
\tfrac12A a^{-1}(\partial_rV)\Theta A
-\tfrac12a^{-1}(\partial_rV)\Theta 
&
=
\tfrac12\mathop{\mathrm{Im}}\bigl(a\bigl(\partial_r a^{-3}(\partial_rV)\Theta\bigr)A\bigr)
\\&\phantom{{}={}}{}
-\tfrac12a^{-3}(\partial_rV)\Theta L
-a^{-3}(\partial_rV)q_1\Theta 
\\&\phantom{{}={}}{}
+\mathop{\mathrm{Re}}\bigl( a^{-3}(\partial_rV)\Theta (H-z)\bigr)
\\&
\ge 
-\tfrac12a^{-3}(\partial_rV)\Theta L
-\C Q
-\D{0}R^{-1}\chi_1
\\&\phantom{{}={}}{}
+\mathop{\mathrm{Re}}\bigl( a^{-3}(\partial_rV)\Theta (H-z)\bigr)
.
\end{split}
\label{eq:2311202022}
\end{align}
The eighth term of \eqref{eq:23080416} requires a technical treatment as follows. 
We first use the Cauchy--Schwarz inequality and Lemma~\ref{lem:theta-inequality} as 
\begin{align}
\begin{split}
\mp\mu\mathop{\mathrm{Re}}(A\Theta )
&\ge 
-\C \mu
-\D{0} \mu p_i\bar \chi_0a^{-2}p_i
\\&
\ge 
-\C\mu
-2\D{1} \mu \mathop{\mathrm{Re}}\bigl(\bar\chi_0a^{-2}\theta (H-z)\bigr)
\\&
\ge 
\mp \D{0}\mathop{\mathrm{Re}}(\mathrm i(H-z))
-\C \mu^2\bar\chi_0r^{\nu/2-1}a^{-2}\theta
\\&\phantom{{}={}}{}
-\D{0}(H-\bar z)\tau(H-z)
.
\end{split}
\label{eq:23111321}
\end{align}
The first and third terms on the right-hand side of \eqref{eq:23111321} are negligible, see the last two terms of \eqref{eq:181218b}.
The second term of \eqref{eq:23111321} is bounded as follows: 
\begin{align}
\begin{split}
-\D{0}\mu^2\bar\chi_0r^{\nu/2-1}a^{-2}\theta
&=
\pm \tfrac12\D{0}\mu\mathop{\mathrm{Re}}\bigl(\bigl(\partial_r\bar\chi_0r^{\nu/2-1}a^{-2}\theta\bigr)p_r\bigr)
\\&\phantom{{}={}}{}
\pm \D{0}\mu\mathop{\mathrm{Im}}\bigl(\bar\chi_0r^{\nu/2-1}a^{-2}\theta(H-z)\bigr)
\\&
\ge 
-\C\mu
-\D{0}\mu p_i\bar\chi_0 r^{2\nu-4}a^{-2}\theta p_i
-\D{0}Q
-\D{0}R^{-1}\chi_1
\\&\phantom{{}={}}{}
-\D{0}\mu^2\bar\chi_0r^{\nu-2}a^{-2}\theta
- \D{0}(H-\bar z)\tau(H-z)
\\&\ge 
-\C \mu
-2C_9\mu \mathop{\mathrm{Re}}\bigl(\bar\chi_0 r^{2\nu-4}a^{-2}\theta(H-z)\bigr)
\\&\phantom{{}={}}{}
-\D{0}Q
-\D{0}R^{-1}\chi_1
\\&\phantom{{}={}}{}
-\D{1}\mu^2\bar\chi_0r^{\nu-2}a^{-2}\theta
- \D{1}(H-\bar z)\tau(H-z)
\\&\ge 
\mp \D{0}\mathop{\mathrm{Re}}(\mathrm i(H-z))
-\D{0}Q
-\D{0}R^{-1}\chi_1
\\&\phantom{{}={}}{}
-\C\mu^2\bar\chi_0r^{\nu-2}a^{-2}\theta
-\D{0}(H-\bar z)\tau(H-z)
.
\end{split}
\label{eq:2311132132}
\end{align}
Note the fourth term on the right-hand side of \eqref{eq:2311132132} has a better decay rate than on the left-hand side, 
or the second term on the right-hand side of \eqref{eq:23111321}, and all the other terms of \eqref{eq:2311132132} are negligible. 
Thus, by repeating this procedure, the second term of \eqref{eq:23111321} can get any extra decay rates up to negligible errors,
so that we obtain 
\begin{align}
\begin{split}
\mp\mu\mathop{\mathrm{Re}}(A\Theta )
&\ge 
\C \mathop{\mathrm{Re}}(\mathrm i(H-z))
-\C Q
\\&\phantom{{}={}}{}
-\D{0}R^{-1}\chi_1
-\D{0}(H-\bar z)\tau(H-z)
.
\end{split}
\label{eq:2311202023}
\end{align}
Hence, by \eqref{eq:23080416}, \eqref{eq:2311202021}, \eqref{eq:2311202022} and \eqref{eq:2311202023}
\begin{align}
\begin{split}
\mathop{\mathrm{Im}}(A\Theta (H-z))
&
\ge 
\tfrac12 A\Theta' A
+\tfrac12 r^{-1}a^{-1}\Theta L
-\tfrac12a^{-3}(\partial_rV)\Theta L
-\C Q
\\&\phantom{{}={}}{}
-\D{0}R^{-1}\chi_1
+\mathop{\mathrm{Re}}(\gamma_1(H-z))
-\D{0}(H-\bar z)\tau(H-z)
,
\end{split}
\label{eq:230923}
\end{align}
where $\gamma_1$ is a bounded function uniformly in $z$ and $R$.

Now we deduce the main positive contributions from the first to third terms of \eqref{eq:230923}. 
By using Assumption~\ref{cond:230806b}, we bound them below as 
\begin{align*}
&
\tfrac12 A\Theta' A
+\tfrac12 r^{-1}a^{-1}\Theta L
-\tfrac12a^{-3}(\partial_rV)\Theta L
\\&
\ge 
\tfrac12 A\Theta' A
+\tfrac12 r^{-1}a^{-3}(2\lambda-\epsilon V)\Theta L
\\&
\ge
\tfrac12 A\Theta' A
+\tfrac{\epsilon}4 r^{-1}a^{-1}\Theta L
\\&
\ge
2c_2 A\Theta' A
+c_2 a^{-2}\Theta' L
+c_2 \tau^{-1}a^{-2}\Theta L
-\C Q
\\&
\ge
c_2 \Theta' 
+c_2 A\Theta' A
+c_2 a^{-2}\tau^{-1}\Theta L
-\C  Q
\\&\phantom{{}={}}{}
-\D{0}R^{-1}\chi_1
-\D{0} \mu
+2c_2 \mathop{\mathrm{Re}}\bigl(a^{-2} \Theta' (H-z)\bigr)
\\&
\ge
c_2 \theta' 
+c_2 p_r^*a^{-2}\theta' p_r
+c_2 a^{-2}\langle \tau\rangle^{-1}\theta L
-\C  Q
\\&\phantom{{}={}}{}
-\D{0}R^{-1}\chi_1
\mp \D{1} \mathop{\mathrm{Re}}(\mathrm i(H-z))
+2c_2 \mathop{\mathrm{Re}}\bigl(a^{-2} \Theta' (H-z)\bigr)
.
\end{align*}
Here we also bound $Q$ as
\begin{align*}
Q
&
\le 
\C r^{-1-\nu'/2}\Theta
+\D{0} R^{-1}\chi_2
+2\mathop{\mathrm{Re}}\bigl(a^{-2}r^{-1-\nu'/2}\Theta (H-z)\bigr)
\\&\phantom{{}={}}{}
+2R^{-1} \mathop{\mathrm{Re}}(\chi_1(H-z))
\\&
\le 
\D{0} r^{-1-\nu'/2}\bar\chi_{n}\theta
+\C R^{-1}\chi_n
+2\mathop{\mathrm{Re}}\bigl(a^{-2}r^{-1-\nu'/2}\Theta (H-z)\bigr)
\\&\phantom{{}={}}{}
+2R^{-1} \mathop{\mathrm{Re}}(\chi_1(H-z))
.
\end{align*} 
Therefore, \eqref{eq:230923} and the above estimates imply 
\begin{align*}
\mathop{\mathrm{Im}}(A\Theta (H-z))
&
\ge 
c_2 \theta' 
-\D{1} r^{-1-\nu'/2}\bar\chi_{n}\theta
+c_2 p_r^*a^{-2}\theta' p_r
+c_2 a^{-2}\langle\tau\rangle^{-1}\theta L
\\&\phantom{{}={}}{}
-\C R^{-1}\chi_n
+\mathop{\mathrm{Re}}(\gamma_2(H-z))
-\D{0}(H-\bar z)\tau(H-z)
.
\end{align*} 
By letting $n\in\mathbb N_0$ be sufficiently large we obtain \eqref{eq:181218b}, hence the assertion. 
\end{proof}

\subsection{Proof of LAP bounds}

Now we prove Theorem~\ref{thm:2308271915}. 
We combine Proposition~\ref{prop:23112315} and contradiction.

\begin{proof}[Proof of Theorem~\ref{thm:2308271915}]
Fix any $\rho>0$ and $\omega\in (0,\pi)$  as in the assertion.

\smallskip
\noindent
\textit{Step 1.}\ 
Here we assume 
\begin{align}\label{eq:lap-single-bound}
\|\phi\|_{\mathcal B^*(\lambda)} \le \CI\|\psi\|_{\mathcal B(\lambda)};\quad 
\phi=R(z)\psi,\ z\in \Gamma_\pm(\rho,\omega),\ \psi\in \mathcal B(\lambda), 
\end{align}
and deduce the assertion by using \eqref{eq:lap-single-bound}. 
Note all the following estimates are uniform in $ z\in \Gamma_\pm(\rho,\omega)$ and $\psi\in \mathcal B(\lambda)$. 
By Proposition~\ref{prop:23112315} with any $\delta\in (0,(\nu'-\nu)/(2+\nu))$ and \eqref{eq:lap-single-bound}, 
we can find $\C>0$ such that uniformly in $\epsilon\in(0,1)$ and $R\ge 1$
\begin{equation*}
 \|a^{-1}\theta'^{1/2}p_r\phi\|^2 
+ \bigl\langle \phi,a^{-2}\langle \tau\rangle^{-1}\theta L\phi\bigr\rangle
\le 
\epsilon\|a^{-1}p_r\phi\|_{{\mathcal B^*(\lambda)}}^2 
+ \epsilon^{-1}\D{0}\|\psi\|_{\mathcal B(\lambda)}^2.
\end{equation*}
Noting the expressions \eqref{eq:15.2.15.5.8} and \eqref{eq:15.2.15.5.9} from $\theta$ and $\theta'$, respectively, 
take the supremum of each term on the above left-hand side in $R\ge 1$, and we obtain
\begin{equation*}
\|a^{-1}p_r\phi\|_{{\mathcal B^*(\lambda)}}^2 
+ \bigl\langle \phi,a^{-2}\langle \tau\rangle^{-1} L\phi\bigr\rangle
\le 
\epsilon \C\|a^{-1}p_r\phi\|_{{\mathcal B^*(\lambda)}}^2 
+\epsilon^{-1}\D{0} \|\psi\|_{\mathcal B(\lambda)}^2.
\end{equation*}
Therefore, by letting $\epsilon\in(0, \D{0}^{-1})$ it follows that
\begin{equation*}
\|a^{-1}p_r\phi\|_{{\mathcal B^*(\lambda)}}^2 
+ \bigl\langle \phi,a^{-2}\langle \tau\rangle^{-1} L\phi\bigr\rangle
\le 
\C\|\psi\|_{\mathcal B(\lambda)}^2.
\end{equation*}
Hence, the assertion reduces to the single bound \eqref{eq:lap-single-bound}.

\smallskip
\noindent
\textit{Step 2.}\ 
Next we prove \eqref{eq:lap-single-bound} by contradiction. 
Let us discuss only the upper sign. 
Assume there exist $z_k\in \Gamma_+(\rho,\omega)$ and $\psi_k\in{\mathcal B(\lambda)}$ such that 
\begin{equation}\label{eq:psi_k}
\lim_{k \to \infty}\|\psi_k\|_{\mathcal B(\lambda)}=0, 
\quad 
\|\phi_k\|_{{\mathcal B^*(\lambda)}}=1;\ \ 
\phi_k=R(z_k)\psi_k.
\end{equation}
By choosing a subsequence, 
we may let $z_k$ converge to some $z\in \overline{\Gamma_+(\rho,\omega)}$ as $k\to\infty$.
If $\mathop{\mathrm{Im}}z>0$, then 
\eqref{eq:psi_k} contradicts the bounds
\begin{equation*}
\|\phi_k\|_{{\mathcal B^*(\lambda)}} 
\le 
\C \|R(z_k)\psi_k\| 
\le 
\|R(z_k)\|_{\mathcal L(\mathcal H)}\|\psi_k\|
\le 
\C\|R(z_k)\|_{\mathcal L(\mathcal H)}\|\psi_k\|_{\mathcal B(\lambda)}
\end{equation*}
and the norm continuity of $R(z)\in \mathcal L(\mathcal H)$ in $z\in \rho(H)$. 
Thus, we have a  real limit
\begin{equation}\label{eq:lim-z_k}
\lim_{k\to\infty}z_k = z = \lambda\in [0,\rho].
\end{equation}

Fix any $s\in (1/2,(2+\nu')/(2(2+\nu))$.
By choosing a subsequence again, we may further let $\langle \tau\rangle^{-s}\phi_k\in \mathcal H$ converge weakly to some 
$\langle \tau\rangle^{-s}\phi \in \mathcal H$. 
Then, in fact, $\langle \tau\rangle^{-s}\phi_k$ converges strongly to $\langle \tau\rangle^{-s}\phi$ in $\mathcal H$.
To see this, take any $t\in(1/2,s)$ and $f\in C_{\mathrm c}^\infty(\mathbb R)$ 
with $f=1$ on a neighborhood of $[0,\rho]$, and decompose $\langle \tau\rangle^{-s}\phi_k$ for any $m\in\mathbb N_0$ as 
\begin{align}
\begin{split}
\langle \tau\rangle^{-s}\phi_k 
&= 
(\langle \tau\rangle^{-s}f(H))(\chi_m \langle \tau\rangle^s)(\langle \tau\rangle^{-s}\phi_k) 
\\&\phantom{{}={}}{}
+ (\langle \tau\rangle^{-s}f(H)\langle \tau\rangle^s)(\bar\chi_m \langle \tau\rangle^{t-s})(\langle \tau\rangle^{-t}\phi_k) 
\\&\phantom{{}={}}{}
+ \langle \tau\rangle^{-s}(1-f(H))R(z_k)\psi_k.
\end{split}
\label{eq:18121814}
\end{align}
By \eqref{eq:psi_k} the last term on the right-hand side 
of \eqref{eq:18121814} converges to $0$ in $\mathcal H$.
Since $\langle \tau\rangle^{-s}f(H)\langle \tau\rangle^s$ is a bounded operator on $\mathcal H$, 
by choosing $m\in\mathbb N_0$ 
sufficiently large the second term of \eqref{eq:18121814}
can be arbitrarily small in $\mathcal H$.
Lastly, since $\langle\tau\rangle^{-s}f(H)$ is a compact operator on $\mathcal H$, 
for any fixed $m\in\mathbb N_0$ the first term of \eqref{eq:18121814} converges strongly in $\mathcal H$.
Therefore, $\langle\tau\rangle^{-s}\phi_k$ is a Cauchy sequence in $\mathcal H$ and converges there to $\langle\tau\rangle^{-s}\phi$:
\begin{equation}\label{eq:lim-phi_k}
\lim_{k\to\infty}\langle\tau\rangle^{-s}\phi_k 
= \langle\tau\rangle^{-s}\phi \ \text{ in } \mathcal H.
\end{equation}

By \eqref{eq:psi_k}, \eqref{eq:lim-z_k} and \eqref{eq:lim-phi_k}, it follows that
\begin{equation}\label{eq:1810031853}
(H-\lambda)\phi = 0 \ \text{ in the distributional sense}.
\end{equation}
In addition, we can verify $\phi\in{\mathcal B_0^*(\lambda)}$.
In fact, letting $\delta\in (2s-1,(\nu'-\nu)/(2+\nu))$, 
we apply Proposition~\ref{prop:23112315} to $\phi_k$ and let $k\to\infty$.
Then by Lemma~\ref{lem:theta-inequality}, \eqref{eq:lim-phi_k} and \eqref{eq:psi_k}, we obtain for all $R\ge 1$
\begin{align}\label{eq:1810031903}
\|\theta'^{1/2}\phi\| 
\le 
\C R^{-1}\|\chi_n^{1/2}\phi\|.
\end{align}
By letting $R\to\infty$ in \eqref{eq:1810031903}, 
we obtain $\phi\in{\mathcal B_0^*(\lambda)}$.
Therefore, by \eqref{eq:1810031853} and Theorem~\ref{thm:230607}
it follows that $\phi=0$

Now we have a contradiction.
In fact, as in Step 1, we can show
\begin{align*}
1=\|\phi_k\|^2_{\mathcal B^*(\lambda)}
\le \C\bigl(\|\psi_k\|_{\mathcal B(\lambda)}^2+\|\chi_n^{1/2}\phi_k\|^2\bigr),
\end{align*}
but the right-hand side can be made smaller than $1$ by taking $k$ large enough.
\end{proof}

\section{Radiation condition bounds}\label{2411245}
In this section, we discuss the radiation condition bounds and their relevant consequences.
Throughout the section, we prove the statements only for the upper sign for simplicity.
\subsection{Commutator estimate}
\begin{lemma}\label{lem:2312221340}
  Let $b$ be the function defined by (\ref{eq:2312221350}), and let $\rho >0$ and $\omega \in (0,\pi)$.
  Then there exists $C>0$ such that uniformly in $z \in \overline{\Gamma_{\pm}(\rho,\omega)}$
  \begin{align*}
    &|b| \leq C,\quad \Im{b} \geq -C\langle x \rangle^{-1-\nu} a^{-2}, \quad |p_{r}b +b^{2} -2(z-V)| \leq C \langle x \rangle^{-2}. 
  \end{align*}
\end{lemma}
\begin{proof}
  It is clear from the definition (\ref{eq:2312221350}) that the first and second inequalities hold.
  Since we can write, 
  \begin{align*}
    p_{r}b +b^{2} -2(z-V) &= -\tfrac{\partial^{2}_{r}V}{4(z-V)}-\tfrac{5(\partial_{r}V)^{2}}{16(z-V)^{2}}.
  \end{align*}
  Hence, the last bound is also clear.
\end{proof}

To simplify a commutator computation in Lemma \ref{lem:2312221423}, we introduce $B$ as
\begin{equation*}
  B = \Re{p_{r}} = p_{r} -\tfrac{i}{2} \Delta r=p_{r} -\tfrac{i(d-1)}{2r} ,
\end{equation*} 
and decompose $H-z$ into a sum of radial and spherical components.

\begin{lemma}\label{lem:2312221420}
  Let $\rho >0$ and $\omega \in (0, \pi)$. Then there exists a complex-valued function $q_{2}$
  and a constant $C>0$ such that uniformly in $z \in \overline{\Gamma_{\pm}(\rho,\omega)} $ 
  \begin{equation*}
    H-z = \tfrac{1}{2}(B+b)(B-b) +\tfrac{1}{2}L + q_{2}\quad  \text{on} \quad \mathbb{R}^{d} \backslash \{0\}
  \end{equation*}
  where $q_{2} \in L^{\infty}(\mathbb{R}^{d})$ satisfies that there exists $C>0$ such that 
  for any $x \in \mathbb{R}^{d}$
  \begin{equation*}
    |q_{2}(x)| \leq C \langle x \rangle^{-1-\nu^{\prime}/2}.
  \end{equation*}  
\end{lemma}
\begin{proof}
  Using the expression \eqref{eq:230825251}, we can write
  \begin{align*}
  H-z &= \tfrac{1}{2}(B-b)(B+b) +\tfrac{1}{2}L + q \\
      & \quad +\tfrac{1}{2}(p_{r}b +b^{2}-2(z-V)) +\tfrac{1}{4} \partial_{r}(\Delta r) +\tfrac{1}{8}(\Delta r)^{2}. 
  \end{align*}
  Hence, we are done.
\end{proof}

Let us introduce the weight
\begin{equation*}
  \Theta = \Theta^{\delta}_{R,l} = \bar{\chi}_{l} \theta;\quad \delta>0,\ \ R\ge1, \ \ l =1,2,3,\ldots.
\end{equation*}
where $\bar{\chi}_{l}$ is a function defined in \eqref{eq:11.7.11.5.14} and $\theta$ is  
\begin{align*}
  \theta
  =\int_0^{\tau/R}(1+s)^{-1-\delta}\,\mathrm ds
  =\delta^{-1}\bigl[1-(1+\tau/R)^{-\delta}\bigr].
\end{align*}

Next we state and prove the key commutator estimate of the section 
that is needed for our proof of the radiation condition bounds.
\begin{lemma}\label{lem:2312221423}
  Let $\rho>0$, $\omega \in (0, \pi)$ and fix any $\delta \in \left(0,\min\left\{\frac{2-\nu}{2(2+\nu)},\frac{\nu^{\prime}-\nu}{2+\nu}\right\} \right)$ and $\beta \in (0,\beta_{c,\rho})$.
  Then there exists $C >0, l \geq 1$ such that for any $R\geq 1$, $z =\lambda+i\mu \in \Gamma_{+}(\rho,\omega)$ 
  and $\psi \in C^{\infty}_{0}(\mathbb{R}^{d})$ the state $\phi= R(z)\psi$ satisfies 
  \begin{equation*}
    \begin{split}
      &\|(\bar{\chi}_{l}\theta^{\prime})^{1/2} \Theta^{\beta-1/2}a^{-1}(p_{r}-b_{z})\phi\|^{2} 
      + \langle\phi,\tau^{2\beta} a^{-2} \langle \tau \rangle^{-1} L\phi \rangle \\
      &\quad \leq C(\|\Theta^{\beta} a^{-1}(p_{r}-b_{z})\phi\|_{\mathcal{B}^{*}(\lambda)} \|\theta^{\beta}\psi\|_{\mathcal{B}(\lambda)}
      +\|\langle \tau \rangle^{-1/2-\min\left\{\frac{2-\nu}{2(2+\nu)},\frac{\nu^{\prime}-\nu}{2+\nu}\right\}+\delta}  \theta^{\beta} \phi\|^{2} \\
      &\quad + \|\langle \tau \rangle^{1/2-\min\left\{\frac{2-\nu}{2(2+\nu)},\frac{\nu^{\prime}-\nu}{2+\nu}\right\} +\delta} \theta^{\beta}\psi\|^{2} ).
    \end{split}
  \end{equation*} 
\end{lemma}
\begin{proof}
  Fix $\rho,\omega,\delta,\beta_{c,\rho}$ in the assertion. 
  Clearly it suffices to show there exist $l\geq 1$ and $c,C>0$ such that 
  \begin{equation}\label{eq:231223715}
    \begin{split}
      \Im((a^{-1}(B-b))^{*}\Theta^{2\beta}(H-z)) &\geq c (a^{-1}(p_{r}-b))^{*} \theta^{\prime} \bar{\chi}_{l} \Theta^{2\beta-1} (a^{-1}(p_{r}-b)) \\ 
                                           &\quad +c p_{i}  a^{-2}\langle \tau \rangle^{-1} \ell_{i,j}\Theta^{2\beta} p_{j} \\
                                           &\quad -C\langle \tau \rangle^{-1-2\min\{(2-\nu)/2(2+\nu),(\nu^{\prime}-\nu)/(2+\nu)\}+2\delta} \theta^{2\beta} \\
                                           &\quad-\Re(\gamma \theta^{2\beta}(H-z)),
    \end{split}
  \end{equation}
  uniformly $R\geq 1$, $z=\lambda +i\mu \in \Gamma_{+}(\rho,\omega)$,
  where $\gamma$ is a complex-valued function satisfying $|\gamma| \leq C \langle \tau \rangle^{-1-2\min\left\{(2-\nu)/(2(2+\nu)),(\nu^{\prime}-\nu)/(2+\nu)\right\}+2\delta}$.
  In the below, we compute and bound the quantity on the left-hand side of \eqref{eq:231223715}.
  We gather admissible error terms and write them for short as
  \begin{equation*}\label{eq:231223719}
    \begin{split}
      Q&=\langle \tau \rangle^{-1-2\min\{(2-\nu)/2(2+\nu),(\nu^{\prime}-\nu)/(2+\nu)\}+2\delta} \theta^{2\beta} \\
      &\quad + p_{i} \langle \tau \rangle^{-1-2\min\{(2-\nu)/2(2+\nu),(\nu^{\prime}-\nu)/(2+\nu)\}+2\delta} \ell_{i,j}\theta^{2\beta}  p_{j}.    
    \end{split}
  \end{equation*}
  For the moment, the following estimates are all uniform in $R\geq 1$, $z \in \Gamma_{+}(\rho,\omega))$,
  and $l\geq 1$.
  By the Lemma~\ref{lem:2312221420}, we can write
  \begin{equation}\label{eq:231223802}
    \begin{split}
      \Im((a^{-1}(B-b))^{*}\Theta^{2\beta}(H-z)) &= \tfrac{1}{2} \Im((a^{-1}(B-b))^{*}\Theta^{2\beta} (B+b)(B-b)) \\
                                           &\quad +\tfrac{1}{2}\Im(Ba^{-1}\Theta^{2\beta}L) -\tfrac{1}{2}\Im((a^{-1}b)^{*}\Theta^{2\beta}L) \\
                                           &\quad +\Im((a^{-1}(B-b))^{*}\Theta^{2\beta}q_{2}).
    \end{split}
  \end{equation}
  The first term of \eqref{eq:231223802} is computed as
  \begin{equation*}
    \begin{split}
      &\tfrac{1}{2} \Im((a^{-1}(B-b))^{*}\Theta^{2\beta} (B+b)(B-b)) \\
      &\quad \geq \tfrac{1}{2} (a^{-1}(B-b))^{*}\beta \Theta^{2\beta-1}  \bar{\chi}_{l} \theta^{\prime} (a^{-1}(B-b)) \\
      &\qquad + \tfrac{1}{4} (a^{-1}(B-b))^{*} a^{-1}\partial_{r}V \Theta^{2\beta} (a^{-1}(B-b))\\
      &\qquad + \tfrac{1}{2} (a^{-1}(B-b))^{*} (\Im b) a  \Theta^{2\beta} (a^{-1}(B-b)) \\
      &\qquad -C_{1}Q.
    \end{split}
  \end{equation*}
  Furthermore, the following inequality holds
  \begin{equation*}
    \begin{split}
      &(\tfrac{1}{4}a^{-1}\partial_{r}V +\tfrac{1}{2} (\Im b)a ) \Theta^{2\beta} \\
      &=\Bigl(\tfrac{1}{2} \Im(\sqrt{2(z-V)}) - \tfrac{\mu^{2}\partial_{r}V}{4a^{2}((\lambda-V)^{2}+\mu^{2})} \\
      &\quad + (\lambda-\max\{\lambda,0\})\tfrac{\partial_{r}V(\lambda-V)}{4a^{2}((\lambda-V)^{2}+\mu^{2})} \Bigr) a \Theta^{2\beta} \\
      &\geq \tfrac{1}{2}\sqrt{\mu}(1 -C_{2}r^{-1+\nu/2})a\Theta^{2\beta}. 
    \end{split}
  \end{equation*}
 Thus, if we take $l$ large enough, we obtain the following estimate:
\begin{equation*}
  \tfrac{1}{2}
  \Im\bigl( 
    (a^{-1}(B-b))^{*}\Theta^{2\beta}(B+b)(B-b)
  \bigr)
  \ge
  \tfrac{1}{2}\,
  (a^{-1}(B-b))^{*}
  \beta \Theta^{2\beta-1}\bar{\chi}_{l}\theta'
  (a^{-1}(B-b)).
\end{equation*}
   Next we can compute the second and third term of \eqref{eq:231223802} as
   \begin{equation*}
    \begin{split}
      &\tfrac{1}{2}\Im(Ba^{-1}\Theta^{2\beta}L) -\tfrac{1}{2}\Im((a^{-1}b)^{*}\Theta^{2\beta}L) \\
      &\quad \geq \tfrac{1}{2} p_{i} r^{-1}a^{-1}(1- \beta r a^{-1} \tau^{-1} -\tfrac{1}{2}a^{-2}r \partial_{r}V +\tfrac{1}{2}r \Im b) \ell_{i,j}\Theta^{2\beta} p_{j} \\
      &\qquad -C_{3}Q \\
      &\quad \geq \tfrac{1}{2} p_{i} r^{-1}a^{-1}(1- \beta r a^{-1} \tau^{-1} +\tfrac{1}{2}(2-\epsilon)a^{-2}V -\tfrac{1}{4} (r \partial_{r}V)/(\lambda-V)) \ell_{i,j}\Theta^{2\beta} p_{j} \\
      &\qquad-C_{3}Q \\
      &\quad \geq \tfrac{1}{2} p_{i} r^{-1}a^{-1}(1- \beta r a^{-1} \tau^{-1} -\tfrac{1}{4}(2-\epsilon)-\tfrac{1}{8}(2-\epsilon)) \ell_{i,j}\Theta^{2\beta} p_{j} \\
      &\qquad-C_{3}Q.
    \end{split}
   \end{equation*} 
   By the definition of $\beta_{c,\rho}$ and Lemma \ref{lem:230925b}, if we take $l$ large enough, we obtain the following estimate:
   \begin{equation*}
    \begin{split}
      &\tfrac{1}{2}\Im(B a^{-1}\Theta^{2\beta}L) -\tfrac{1}{2}\Im((a^{-1}b)^{*}\Theta^{2\beta}L) \\
      &\geq c_{1} p_{i} a^{-2} \langle \tau \rangle^{-1} \ell_{i,j}\Theta^{2\beta}p_{j }-C_{3}Q.
    \end{split}
  \end{equation*}
   We can directly compute the fourth term of \eqref{eq:231223802} by using the Cauchy--Schwarz inequality as
   \begin{equation*}
    \begin{split}
     &\Im((a^{-1}(B-b))^{*}\Theta^{2\beta}q_{2}) \\
     &\quad \geq -C_{4}(a^{-1}(B-b))^{*} \langle \tau \rangle^{-1-2\delta}(a^{-1}(B-b)) -C_{4}Q. 
    \end{split}
   \end{equation*}
   Thus, it follows that 
    \begin{equation*}
      \begin{split}
         &\Im((a^{-1}(B-b))^{*}\Theta^{2\beta}(H-z))\\
         &\quad \geq  \tfrac{1}{2} (a^{-1}(B-b))^{*}(\beta \bar{\chi}_{l}\theta^{\prime} - C_{5}\langle \tau \rangle^{-1-2\delta}\Theta)\Theta^{2\beta-1} (a^{-1}(B-b)) \\
         &\qquad +c_{1} p_{i}a^{-2} \langle \tau \rangle^{-1} \ell_{i,j}\Theta^{2\beta}p_{j} 
         -C_{5}Q.
      \end{split}
    \end{equation*}
   By Lemma \ref{lem:theta-inequality}, we have
   \begin{equation*}
    \begin{split}
      &\tfrac{1}{2} (a^{-1}(B-b))^{*}(\beta \bar{\chi}_{l}\theta^{\prime} - C_{5}\langle \tau \rangle^{-1-2\delta}\Theta)\Theta^{2\beta-1} (a^{-1}(B-b)) \\
      &\quad \geq c_{2} (a^{-1}(B-b))^{*}\bar{\chi}_{l} \theta^{\prime} \Theta^{2\beta-1} (a^{-1}(B-b)) 
      -C_{6}Q \\
      &\quad \geq c_{3} (a^{-1}(p_{r}-b))^{*}\bar{\chi}_{l} \theta^{\prime} \Theta^{2\beta-1} (a^{-1}(p_{r}-b)) 
      -C_{7}Q.
    \end{split}
   \end{equation*}
Finally we can bound $-Q$ as 
\begin{equation*}
  \begin{split}
    -Q &\geq -C_{8} \langle \tau \rangle^{-1-2\min\{(2-\nu)/2(2+\nu),(\nu^{\prime}-\nu)/(2+\nu)\}+2\delta} \theta^{2\beta} \\
       &\quad -C_{8} \Re(\langle \tau \rangle^{-1-2\min\{(2-\nu)/2(2+\nu),(\nu^{\prime}-\nu)/(2+\nu)\}+2\delta} (H-z) ).
  \end{split}
\end{equation*}
Hence, we are done.
\end{proof}

\subsection{Applications}
Now we are going to prove Theorem \ref{thm:2308271916} and Corollary \ref{cor:2312221102}--\ref{cor:2312221107} in this order.
\subsubsection{Radiation condition bounds for complex spectral parameters}
\begin{proof}[Proof of Theorem \ref{thm:2308271916}]
  Let $\rho,\omega$ be as in the assertion. 
  For $\beta =0$, the assertion is obvious by Theorem \ref{thm:2308271915}, 
  and hence, we may let $\beta \in (0,\beta_{c,\rho})$. 
  We take any
  \begin{equation*}
    \delta \in \left(0,\min\left\{0, \min\left\{\tfrac{2-\nu}{2(2+\nu)},\tfrac{\nu^{\prime}-\nu}{2+\nu}\right\} -\beta_{c,\rho} \right\}\right).
  \end{equation*}
  By Lemma \ref{lem:2312221423}, 
  there exists $C_{1}>0$ such that for any state $\phi = R(z)\psi$ 
  with $\psi \in C^{\infty}_{0}(\mathbb{R}^{d})$ and $z \in \Gamma_{+}(\rho,\omega)$
  \begin{equation}\label{eq:2312231523}
    \begin{split}
      &\|(\bar{\chi}_{l}\theta^{\prime})^{1/2} \Theta^{\beta-1/2}a^{-1}(p_{r}-b_{z})\phi\|^{2} 
      + \langle\phi,\tau^{2\beta} a^{-2} \langle \tau \rangle^{-1} L\phi \rangle \\
      &\quad \leq C(\|\Theta^{\beta} a^{-1}(p_{r}-b_{z})\phi\|_{\mathcal{B}^{*}(\lambda)} \|\theta^{\beta}\psi\|_{\mathcal{B}(\lambda)}
      +\|\langle \tau \rangle^{-1/2-\min\left\{\frac{2-\nu}{2(2+\nu)},\frac{\nu^{\prime}-\nu}{2+\nu}\right\}+\delta} \theta^{\beta} \phi\|^{2} \\
      &\qquad + \|\langle \tau \rangle^{1/2-\min\left\{\frac{2-\nu}{2(2+\nu)},\frac{\nu^{\prime}-\nu}{2+\nu}\right\}+\delta} \theta^{\beta} \psi\|^{2} ) \\
      &\quad \leq C_{1} R^{-2\beta} (\|\bar{\chi}_{l}^{1/2} \tau ^{\beta}a^{-1}(p_{r}-b)\phi\|_{\mathcal{B}^{*}(\lambda)} \|\langle \tau \rangle^{\beta} \psi\|_{\mathcal{B}(\lambda)}
       + \|\langle \tau \rangle^{\beta} \psi\|_{\mathcal{B}(\lambda)}^{2}).
    \end{split}
  \end{equation} 
  Here we note that $\langle \tau \rangle^{\beta} a^{-1}(p_{r}-b)\phi \in \mathcal{B}^{*}$ for each $z \in \Gamma_{+}(\rho,\omega)$,
  and hence, the quantity on the right-hand side of \eqref{eq:2312231523} is finite.
  In fact, this can be verified by commuting $R(z)$ 
  and powers of $\langle \tau \rangle$ sufficiently many times and using the fact that $\psi \in C^{\infty}_{0}(\mathbb{R}^{d})$.
  Then by \eqref{eq:2312231523}, it follows
  \begin{equation}\label{eq:2312231540}
    \begin{split}
      &R^{2\beta} \|(\bar{\chi}_{l}\theta^{\prime})^{1/2} \Theta^{\beta-1/2}a^{-1}(p_{r}-b_{z})\phi\|^{2} 
      + R^{2\beta} \langle\phi,\tau^{2\beta} a^{-2} \langle \tau \rangle^{-1} L\phi \rangle \\ 
      &\quad \leq C_{1}  (\|\bar{\chi}_{l}^{1/2} \tau^{\beta}a^{-1}(p_{r}-b)\phi\|_{\mathcal{B}^{*}(\lambda)} \|\langle \tau \rangle^{\beta} \psi\|_{\mathcal{B}(\lambda)}
      + \|\langle \tau \rangle^{\beta} \psi\|_{\mathcal{B}(\lambda)}^{2}).
    \end{split}
  \end{equation} 
  In the first term on the left-hand side of \eqref{eq:2312231540},
  let $R$ be $2^{m}$ and take the supremum in $m \geq 0$ 
  noting \eqref{eq:15.2.15.5.9}, and then, obtain
  \begin{equation*}\label{eq:2312231600}
    \begin{split}
      &c_{1} \|(\bar{\chi}_{l})^{1/2} \tau^{\beta} a^{-1}(p_{r}-b_{z})\phi\|_{\mathcal{B}^{*}(\lambda)}^{2} \\
      &\quad \leq C_{1}  (\|\bar{\chi}_{l}^{1/2} \tau^{\beta}a^{-1}(p_{r}-b)\phi\|_{\mathcal{B}^{*}(\lambda)} \|\langle \tau \rangle^{\beta} \psi\|_{\mathcal{B}(\lambda)}
      + \|\langle \tau \rangle^{\beta} \psi\|_{\mathcal{B}(\lambda)}^{2}).
    \end{split}
  \end{equation*}
  which implies
  \begin{equation}\label{eq:2312231623}
    \|\bar{\chi}^{1/2}_{l}\tau^{\beta} a^{-1}(p_{r}-b)\phi\|_{\mathcal{B}(\lambda)} 
    \leq C_{2} \|\langle \tau \rangle^{\beta} \psi\|_{\mathcal{B}(\lambda)}.
  \end{equation}
  As for the second term on the left-hand side of \eqref{eq:2312231540}, 
  we use \eqref{eq:2312231623}, the concavity of $\Theta$ 
  and Lebesgue's monotone convergence theorem and then obtain 
  by letting $R \rightarrow \infty$
  \begin{equation}\label{eq:2312231602}
    \langle p_{i} \bar{\chi}_{l}^{1/2} \tau^{2\beta} a^{-2} \langle \tau \rangle^{-1} \ell_{i,j} p_{j}   \rangle_{\phi}
    \leq C_{3}\|\langle \tau \rangle^{\beta} \psi\|_{\mathcal{B}(\lambda)}.
  \end{equation}
  From \eqref{eq:2312231623} and \eqref{eq:2312231602}, we can remove the cut off $\bar{\chi}_{l}^{1/2}$ 
  by using Theorem \ref{thm:2308271915}. Hence, we are done.
\end{proof}

\subsubsection{Limiting absorption principle}
\begin{proof}[Proof of corollary \ref{cor:2312221102}]
Let $s>1/2$ and $\gamma \in (0,\min\{s-1/2,\beta_{c,\rho}\})$ be in the assertion. 
Let $s^{\prime}=s-\gamma$.
We decompose for $n \geq 0$ and $z,z^{\prime} \in \Gamma_{+}(\rho,\omega)$ 
\begin{equation}
  \begin{split}\label{eq:2312222353}
    R(z)-R(z^{\prime}) &= \chi_{n}R(z)\chi_{n} -\chi_{n}R(z^{\prime})\chi_{n} \\
                       &\quad + (R(z)-\chi_{n}R(z)\chi_{n}) -(R(z^{\prime})-\chi_{n}R(z^{\prime})\chi_{n}).  
    \end{split}
\end{equation}
By Theorem \ref{thm:2308271915}, we can estimate the third term of \eqref{eq:2312222353}
uniformly in $n \geq 0$ and $z,z^{\prime} \in \Gamma_{+}(\rho,\omega)$ as
\begin{equation}\label{eq:23122414415}
  \begin{split}
  &\|R(z)-\chi_{n}R(z)\chi_{n}\|_{\mathcal{B}\left(L^{2}_{s,\min\{\Re{z},\Re{z^{\prime}}\}}, L^{2}_{-s,\min\{\Re{z},\Re{z^{\prime}}\}}\right)} \\
  &\leq \|\langle \tau \rangle^{-s} \bar{\chi}_{n} R(z) \bar{\chi}_{n} \langle \tau \rangle^{-s}\|_{\mathcal{B(\mathcal{H})}} 
  + \|\langle \tau \rangle^{-s} \chi_{n} R(z) \bar{\chi}_{n} \langle \tau \rangle^{-s}\|_{\mathcal{B(\mathcal{H})}} \\
  &\quad +\|\langle \tau \rangle^{-s} \bar{\chi}_{n} R(z) \chi_{n} \langle \tau \rangle^{-s}\|_{\mathcal{B(\mathcal{H})}} \\
  &\leq C_{1}2^{n(s^{\prime}-s)}=C_{1}2^{-n\gamma}. 
  \end{split}
\end{equation}  
Similarly, we obtain
\begin{equation*}\label{eq:12241433}
    \|R(z^{\prime})-\chi_{n}R(z^{\prime})\chi_{n}\|_{\mathcal{B}\left(L^{2}_{s,\min\{\Re{z},\Re{z^{\prime}}\}}, L^{2}_{-s,\min\{\Re{z},\Re{z^{\prime}}\}}\right)} 
    \leq C_{2}2^{-n\gamma}. 
\end{equation*}
As for the first and second term 
on the right-hand side of \eqref{eq:2312222353},
using the equation
\begin{equation}
  \mathrm i[H,\chi_{n+1}] =\Re(\chi_{n+1 }^{\prime}a^{-1}p_{r} )
\end{equation}
and noting the identity $\bar{b}_{\bar{z}}=b_{z}$, 
we can write for $n\geq 0$ 
\begin{equation}\label{eq:2312241434}
  \begin{split}
    &\chi_{n}R(z)\chi_{n}-\chi_{n}R(z^{\prime})\chi_{n} \\
    &=\chi_{n}R(z){\chi_{n+1}(H-z^{\prime})-(H-z)\chi_{n+1}}R(z^{\prime})\chi_{n} \\
    &=\chi_{n}R(z)\left\{(z-z^{\prime})\chi_{n+1}+i \Re(\chi^{\prime}_{n+1}a^{-1}p_{r})\right\}R(z^{\prime})\chi_{n} \\
    &=\tfrac{i}{2}\chi_{n}R(z)\chi^{\prime}_{n+1}a^{-1}(p_{r}-b_{z^{\prime}})R(z^{\prime})\chi_{n}
      +\tfrac{i}{2}\chi_{n}R(z)(a^{-1}(p_{r}+b_{\bar{z}}))^{*}\chi^{\prime}_{n+1}R(z^{\prime}) \chi_{n} \\
    &\quad - \tfrac{i}{2}\chi_{n}R(z)(a^{-1}(b_{z}-b_{z^{\prime}}))\chi^{\prime}_{n+1}R(z^{\prime})\chi_{n} 
      -(z-z^{\prime})\chi_{n}R(z)\chi_{m}R(z^{\prime})\chi_{n} \\
    &\quad - (z-z^{\prime}) \chi_{n}R(z)\chi_{m,n+1}(a(b_{z}+b_{z}^{\prime})^{-1})(a^{-1}(p_{r}-b_{z^{\prime}}))R(z^{\prime})\chi_{n} \\
    &\quad  +(z-z^{\prime}) \chi_{n}R(z)(a^{-1}(p_{r}+b_{\bar{z}}))^{*}\chi_{m,n+1} a (b_{z}+b_{z^{\prime}})^{-1} R(z^{\prime}) \chi_{n} \\
    &\quad -(z-z^{\prime}) \chi_{n} R(z) [a^{-1}p_{r}, \chi_{m,n+1}a(b_{z}+b_{z^{\prime}})^{-1}] R(z^{\prime})\chi_{n}.
  \end{split}
\end{equation}
Here $m\geq 0$ is fixed so that $(b_{z}+b_{z^{\prime}})^{-1}$ is non-singular on supp $\bar{\chi}_{m}$.
Then by Theorem \ref{thm:2308271915} and \ref{thm:2308271916}, 
we have uniformly in $n\geq 0$ and $z,z^{\prime} \in \Gamma_{+}(\rho,\omega)$
\begin{equation}\label{eq:2312241436}
  \begin{split}
    &\|\chi_{n}R(z)\chi_{n}-\chi_{n}R(z^{\prime})\chi_{n}\|_{\mathcal{B}\left(L^{2}_{s,\min\{\Re{z},\Re{z^{\prime}}\}}, L^{2}_{-s,\min\{\Re{z},\Re{z^{\prime}}\}}\right)} \\
    &\leq C_{3}2^{-n\gamma}+C_{4}2^{n(1-\gamma)} |z-z^{\prime}|.
  \end{split}
\end{equation}
Summing up \eqref{eq:2312222353}--\eqref{eq:2312241436},
we obtain uniformly in $n \geq 0$ and $z,z^{\prime} \in \Gamma_{+}(\rho,\omega)$ 
\begin{equation*}\label{eq:2312241439}
  \|R(z)-R(z^{\prime})\|_{\mathcal{B}\left(L^{2}_{s,\min\{\Re{z},\Re{z^{\prime}}\}}, L^{2}_{-s,\min\{\Re{z},\Re{z^{\prime}}\}}\right)}
  \leq C_{5}2^{-n\gamma }+C_{5}2^{n(1-\gamma)}|z-z^{\prime}|.
\end{equation*}
Now, if $|z-z^{\prime}| \leq 1$, then we choose $2^{n} \leq |z-z^{\prime}|^{-1} <2^{n+1}$ 
and then obtain
\begin{equation}\label{eq:2312241440}
  \|R(z)-R(z^{\prime})\|_{\mathcal{B}\left(L^{2}_{s,\min\{\Re{z},\Re{z^{\prime}}\}}, L^{2}_{-s,\min\{\Re{z},\Re{z^{\prime}}\}}\right)}
  \leq C_{6}|z-z^{\prime}|^{\gamma}.
\end{equation} 
The same bound is trivial for $|z-z^{\prime}| \geq 1$, 
and hence, the H\"{o}lder continuity \eqref{eq:2312221103}
for $R(z)$ follows from \eqref{eq:2312241440}.
The H\"{o}lder continuity for $a^{-1}p_{r}R(z)$ follows
by using the first resolvent equation.
The existence of the limits \eqref{eq:2312221104} is an immediate consequence of \eqref{eq:2312221103}.
By Theorem \ref{thm:2308271915} the limits $R(z)$ and $a^{-1}p_{r}R(z)$ actually map into $\mathcal{B}^{*}$,
and moreover, they extend continuously to map $\mathcal{B} \rightarrow \mathcal{B}^{*}$ by density argument.
Hence, we are done.
\end{proof}

\subsubsection{Radiation condition bounds for real spectral parameters}
  \begin{proof}[Proof of Corollary~\ref{cor:2312221105}]
   The corollary follows from Theorem \ref{thm:2308271916}, Corollary \ref{cor:2312221102} and approximation arguments.
   Note the elementary property
   \begin{equation*}
    \|\psi\|_{\mathcal{B}^{*}(\lambda)} = \sup_{n\geq 0} \|\chi_{n}\psi\|_{\mathcal{B}^{*}(\lambda)};\quad \psi \in \mathcal{B^{*}(\lambda)}.
   \end{equation*}
   Hence, we are done.
  \end{proof}

\subsubsection{Sommerfeld uniqueness result}
  \begin{proof}[Proof of Corollary~\ref{cor:2312221107}]
    Let $\lambda \geq 0$, $\phi \in L^{2}_{\mathrm{loc}}$ and $\psi \in \langle \tau \rangle^{-\beta} \mathcal{B}(\lambda)$ with $\beta \in [0,\beta_{c,\lambda})$.
    We first assume $\phi = R(\lambda +\mathrm i0)\psi$. Then (i) and (ii) of the corollary hold by Corollaries \ref{cor:2312221102} and \ref{cor:2312221105}.
    Conversely, assume (i) and (ii) of the corollary, and let 
    \begin{equation*}
      \phi^{\prime} = \phi -R(\lambda +\mathrm i0)\psi.
    \end{equation*}
    Then by Corollaries \ref{cor:2312221102} and \ref{cor:2312221105}, it follows that $\phi^{\prime}$ satisfies of the corollary with $\psi =0$.
    In addition, we can verify $\phi^{\prime} \in \mathcal{B}^{*}_{0}(\lambda)$ by the virial-type argument.
    In fact noting the identity 
    \begin{equation*}
      2\Im(\chi_{m}(H-\lambda)) = \Re(a^{-1}b)\chi^{\prime}_{m} + \Re(\chi_{m}^{\prime}a^{-1}(p_{r}-b)),
    \end{equation*}
    we conclude that 
    \begin{equation}\label{eq:2312222241}
      0 \leq \langle \Re(a^{-1}b)\bar{\chi}^{\prime}_{m}\rangle_{\phi^{\prime}} = \Re \langle \chi_{m}^{\prime}a^{-1}(p_{r}-b) \rangle_{\phi^{\prime}}. 
    \end{equation}
    Taking the limit $m \rightarrow \infty$ and using $\phi^{\prime} \in \langle \tau \rangle^{\beta} \mathcal{B}^{*}(\lambda)$ 
    and $a^{-1}(p_{r}-b)\phi^{\prime} \in \langle \tau \rangle^{-\beta} \mathcal{B}^{*}_{0}(\lambda)$ in \eqref{eq:2312222241},
    we obtain $\phi^{\prime} \in \mathcal{B}^{*}_{0}(\lambda)$. 
    By Theorem \ref{thm:230607}, it follows that $\phi^{\prime}=0$. Hence we have $\phi = R(\lambda+\mathrm i0)\psi$. 
  \end{proof}

\subsubsection*{Acknowledgements} 
KI thanks Erik Skibsted for valuable discussions and comments on the draft. 
The authors thank Akitoshi Hoshiya for some references.

\subsubsection*{Funding}
KI was partially supported by JSPS KAKENHI Grant Number JP23K03163.

\subsubsection*{Data Availability}
Data sharing is not applicable to this article as no datasets were generated or analyzed during the current study.

\section*{Declarations}

\subsubsection*{Conflict of interest}
The authors declare that they have no conflict of interest.

\appendix 

\section{LAP bounds in high-energy regime}\label{sec:231120}

\subsection{Settings and results}

In this appendix, we discuss the LAP bounds in the high-energy regime.
The settings here are slightly more general than in the previous sections, 
being rather independent of them, for there do not appear the difficulties due to the energy $0$. 
See also Yafaev's book~\cite{Y2} for related results.  

Consider the Schr\"odinger operator of the form 
\[
H=\tfrac12 p^2+V
\]
with the following assumption on $V$: 
\begin{assumption}\label{cond:23112420}
Let $V\in L^\infty(\mathbb R^d;\mathbb R)$, 
and there exists a splitting 
  \begin{align*}
  V=V_1+V_2;\quad V_1\in C^1(\mathbb R^d;\mathbb R),\ \ V_2\in L^\infty(\mathbb R^d;\mathbb R),
  \end{align*}
such that for some $C>0$ and $\epsilon\in (0,1)$ 
\begin{align*}
\partial_r V_1\le C\langle x\rangle^{-1-\epsilon},\quad
|V_2|\le C\langle x\rangle^{-1-\epsilon}.
\end{align*}
\end{assumption}
\begin{remark}
Clearly, this is weaker than Assumption~\ref{cond:230806}. 
\end{remark}

We also slightly refine the Agmon--H\"ormander spaces
 as, for any $m\in\mathbb N_0$, 
\begin{align*}
\mathcal B(m)&=
\bigl\{\psi\in L^2_{\mathrm{loc}};\ \|\psi\|_{\mathcal B(m)}<\infty\bigr\},\quad 
\|\psi\|_{\mathcal B(m)}=\sum_{n\in\mathbb N} 2^{(n-m)/2}\|1_n(m)\psi\|_{\mathcal H},\\
\mathcal B^*(m)&=
\bigl\{\psi\in L^2_{\mathrm{loc}};\ \|\psi\|_{\mathcal B^*(m)}<\infty\bigr\},\quad 
\|\psi\|_{\mathcal B^*(m)}=\sup_{n\in\mathbb N}2^{-(n-m)/2}\|1_n(m)\psi\|_{\mathcal H},
\end{align*}
with 
\begin{align*}
\begin{split}
1_1(m)&=1\bigl(\bigl\{x\in\mathbb R^d;\ |x|<2^{1-m}\bigr\}\bigr),\\
1_n(m)&=1\bigl(\bigl\{x\in\mathbb R^d;\ 2^{n-m-1}\le |x|<2^{n-m}\bigr\}\bigr)
\ \  \text{for }n=2,3,\ldots .
\end{split}
\end{align*}
Note $\mathcal B(m)$ coincide for all $m\in\mathbb N_0$ as sets, and so do $\mathcal B^*(m)$,
however, we have 
\begin{align}
\|\cdot\|_{\mathcal B(m)}\ge\|\cdot\|_{\mathcal B(m+1)},\quad 
\|\cdot\|_{\mathcal B^*(m)}\le \|\cdot\|_{\mathcal B^*(m+1)}.
\label{eq;23112417}
\end{align}
For any $a,b >0$, we set 
\[
\Gamma_{\pm}'(a,b)
=\bigl\{z=\lambda\pm\mathrm i\mu \in\mathbb C;\ \lambda>a,\ 0<\mu<b\lambda^{1/2}\bigr\}, 
\]
respectively.

\begin{theorem}\label{thm:2308271915b}
Suppose Assumption~\ref{cond:23112420}.
Let $a>0$ be sufficiently large, and let $b>0$. 
Then there exists $C>0$ such that for any $z=\lambda\pm\mathrm i\mu\in \Gamma_{\pm}'(a,b)$, 
$\psi\in \mathcal B(\lambda)$ and $m\in\mathbb N_0$ with $2^{2m}<\lambda\le 2^{2m-2}$
\begin{align}
\|R(z)\psi\|_{\mathcal B^*(m)}
&\le C2^{-m}\|\psi\|_{\mathcal B(m)},
\label{eq:23112513}
\\
\|p_rR(z)\psi\|_{\mathcal B^*(m)}
&\le C\|\psi\|_{\mathcal B(m)},
\notag
\\
\bigl\langle R(z)\psi,\langle 2^{m}x\rangle^{-1}LR(z)\psi\bigr\rangle
&\le C2^{-m}\|\psi\|_{\mathcal B(m)}^2,
\notag
\end{align}
respectively. 
\end{theorem}
\begin{remark}
According to \eqref{eq;23112417},
this refines the ordinary $\mathcal B(0)$-$\mathcal B^*(0)$ estimates in the high-energy regime, 
where the Agmon--H\"ormander spaces approach those with homogeneous weights. 
\end{remark}

\subsection{Outline of proof}

\subsubsection{Reduction}

The proof of Theorem~\ref{thm:2308271915b} depends on a simple scaling argument
combined with the following local resolvent estimates for an $h$-dependent Schr\"odinger operator 
\[
H_h=\tfrac12p^2+V_h
;\quad 
V_h(x)=h^2V(hx),\ \ 
h\in (0,1].
\]
Let $R_h(w)=(H_h-w)^{-1}$ for $w\in\rho(H_h)$, and also set for any $b >0$ 
\[
\Gamma_{\pm}''(b)
=\bigl\{z=\lambda\pm\mathrm i\mu \in\mathbb C;\ 1< \lambda\le 4,\ 0<\mu<b\bigr\}, 
\]
respectively.

\begin{proposition}\label{prop:23112415}
Suppose Assumption~\ref{cond:23112420}.
Let $h_0>0$ be sufficiently small, and let $b>0$. 
Then there exists $C>0$ such that for any $h\in(0,h_0]$, $w\in \Gamma''_\pm(b)$ 
and $\psi\in \mathcal B(0)$
\[
\|R_h(w)\psi\|_{\mathcal B^*(0)}
+\|p_rR_h(w)\psi\|_{\mathcal B^*(0)}
+\bigl\langle R_h(w)\psi,\langle x\rangle^{-1}LR_h(w)\psi\bigr\rangle^{1/2}
\le C\|\psi\|_{\mathcal B(0)}
.
\]
\end{proposition}

The proof of Proposition~\ref{prop:23112415} is essentially the same as \cite[Proposition~4.1]{IS}
and hence is very similar to Theorem~\ref{thm:2308271915}. 
We will only present key steps later on. 
Here let us verify Theorem~\ref{thm:2308271915b} by using Proposition~\ref{prop:23112415} and scaling.

\begin{proof}[Deduction of Theorem~\ref{thm:2308271915b} from Proposition~\ref{prop:23112415}]
We shall prove only \eqref{eq:23112513} since the others are treated similarly. 
For any $m\in\mathbb N_0$ and $\psi\in L^2_{\mathrm{loc}}(\mathbb R^d)$, we set 
\[
(U_m\psi)(x)=\psi(2^{-m} x). 
\]
Then by definition of the norms in $\mathcal B(m)$ and $\mathcal B^*(m)$
and change of variables, it is straightforward to see that 
\[
\|\phi\|_{\mathcal B^*(m)}
=\|U_m\phi\|_{\mathcal B^*(0)}
,\quad 
\|\psi\|_{\mathcal B(m)}
=2^{-m}\|U_m\psi\|_{\mathcal B(0)}
,
\]
so that 
\[
\|\phi\|_{\mathcal B^*(m)}
\le 
2^m\|U_mR(z)U_m^{-1}\|_{\mathcal L(\mathcal B(0),\mathcal B^*(0))}\|\psi\|_{\mathcal B(m)}
.
\]
Hence, it suffices to examine $U_mR(z)U_m^{-1}\colon \mathcal B(0)\to\mathcal B^*(0)$. 
Let us rewrite 
\[
U_mR(z)U_m^{-1}=2^{-2m}R_h (w)
\ \ \text{with }h=2^{-m},\ w=2^{-2m}z
.
\]
If we choose $m\in\mathbb N_0$ such that $2^{2m}<\mathop{\mathrm{Re}}z\le 2^{2m+2}$, 
then Proposition~\ref{prop:23112415} indeed applies, and the assertion follows. 
We are done. 
\end{proof}

\subsubsection{$h$-Resolvent estimates}

In order to prove Proposition~\ref{prop:23112415}, we employ the following estimates with
\begin{align*}
\theta
=\int_0^{r/R}(1+s)^{-1-\delta}\,\mathrm ds
=\delta^{-1}\bigl[1-(1+r/R)^{-\delta}\bigr]
;\quad \delta>0,\ \ R\ge1,
\end{align*}
cf.\ \eqref{eq:15.2.15.5.8}.

\begin{lemma}\label{lem:23112315b}
Let $h_0>0$ be sufficiently small, and let $b>0$ and $\delta\in (0,\epsilon)$. 
Then there exist $C>0$ and $n \in\mathbb N_0$ such that 
for any $R\ge 1$, $h\in(0,1]$, $w\in \Gamma_{\pm}''(b)$ and $\psi \in \mathcal B(0)$
the state $\phi=R_h(w)\psi$ satisfies
\begin{align*}
&
 \|\theta'^{1/2}\phi\|^2 
+  \|\theta'^{1/2}p_r\phi\|^2 
+ \bigl\langle \phi,\langle x\rangle^{-1}\theta L\phi\bigr\rangle
\\&\le 
C\left(
\|\phi\|_{\mathcal B^*(0)} \|\psi\|_{\mathcal B(0)} 
+ \|p_r\phi\|_{\mathcal B^*(0)} \|\psi\|_{\mathcal B(0)} 
+ R^{-1}\|\chi_n^{1/2}\phi\|^2 \right).
\end{align*}
\end{lemma}
\begin{proof}
Similarly to the proof of  \cite[Proposition~4.1]{IS} or Theorem~\ref{thm:2308271915}, 
it suffices to compute and bound from below 
a distorted commutator 
\[
\mathop{\mathrm{Im}}(A\Theta (H_h-w))
;\quad 
A=\mathop{\mathrm{Re}}p_r,\ \  \Theta=(1-\chi(r))\theta, 
\]
where $\chi$ is from \eqref{eq:14.1.7.23.24}. 
This can be done in almost the same manner as for \cite[Proposition~4.1]{IS},
which is indeed much simpler than Theorem~\ref{thm:2308271915}. 
The only differences from \cite[Proposition~4.1]{IS} are that 
our $r$ is singular at the origin, and that we have an $h$ dependence. 
However, the former is handled by the cutoff  $1-\chi(r)$. 
As for the latter, note that $V_h$ satisfies Assumption~\ref{cond:23112420}, 
see also \cite[Condition~1.3]{IS}, uniformly 
in $h\in (0,1]$, and moreover, that the corresponding coefficient $C$ vanishes as $h\to +0$. 
Hence, the assertion can be verified for any small $h\in(0,1]$.
Let us omit the further details. 
\end{proof}

\begin{proof}[Proof of Proposition~\ref{prop:23112415}]
Once we obtain Lemma~\ref{lem:23112315b}, 
Proposition~\ref{prop:23112415} follows from it by 
the same manner as for Theorem~\ref{thm:2308271915}, or \cite[Theorem~1.7]{IS}. 
Here we have an $h$ dependence, but a modification is straightforward. 
We omit the details. 
\end{proof}

\end{document}